\documentclass[a4paper,UKenglish]{lipics-v2019}

\usepackage[toc,page]{appendix}

\usepackage[utf8]{inputenc}
\usepackage[gen]{eurosym}

\usepackage{subcaption}
\usepackage[shortlabels]{enumitem}
\usepackage{float}

\usepackage{tikz}
\usetikzlibrary{arrows,shapes}

\hideLIPIcs
\nolinenumbers

\makeatletter
\@fleqnfalse
\makeatother

\title{Sequential Defaulting in Financial Networks}

\author{Pál András Papp}{ETH Zürich, Switzerland}{apapp@ethz.ch}{}{}
\author{Roger Wattenhofer}{ETH Zürich, Switzerland}{wattenhofer@ethz.ch}{}{}

\authorrunning{P.A. Papp and R. Wattenhofer}

\Copyright{Pál András Papp and Roger Wattenhofer}



\begin{document}

\maketitle

\begin{abstract}
We consider financial networks, where banks are connected by contracts such as debts or credit default swaps.
We study the clearing problem in these systems: we want to know which banks end up in a default, and what portion of their liabilities can these defaulting banks fulfill. We analyze these networks in a sequential model where banks announce their default one at a time, and the system evolves in a step-by-step manner.

We first consider the \textit{reversible model} of these systems, where banks may return from a default.
We show that the stabilization time in this model can heavily depend on the ordering of announcements. However, we also show that there are systems where for any choice of ordering, the process lasts for an exponential number of steps before an eventual stabilization. We also show that finding the ordering with the smallest (or largest) number of banks ending up in default is an NP-hard problem. Furthermore, we prove that defaulting early can be an advantageous strategy for banks in some cases, and in general, finding the best time for a default announcement is NP-hard. Finally, we discuss how changing some properties of this setting affects the stabilization time of the process, and then use these techniques to devise a \textit{monotone model} of the systems, which ensures that every network stabilizes eventually.
\end{abstract}

\vspace{36pt}

\section{Introduction} \label{sec:Intro}

The world's financial system is a highly complex network where banks and other financial institutions are interconnected by various kinds of contracts. These connections create a strong interdependence between the banks: if one of them goes bankrupt, then this also affects others, causing a cascading effect through the network. Such ripple effects also had an important role in the financial crisis of 2008, and hence there is an increasing interest in the network-based properties of these systems.

One fundamental question in these networks is the so-called \textit{clearing problem}: given a network of banks and contracts, we need to decide which of the banks can fulfill their payment obligations, and which of the banks cannot, and thus have to report a default. This question is of high interest both for financial authorities and for the banks involved.

With two simple kinds of contracts, one can already build a financial network model that captures a wide range of phenomena in real-life financial systems. Previous work has mostly focused on the equilibrium states in these models, i.e. the fixed final states where the recovery rates of banks are consistent with their current assets and liabilities. However, in practice, most events in a financial system happen gradually, one after another: a single bank announces a default, which might prompt another bank to reevaluate its situation, and also need to call being in default. This sequential development is an inherent part of the way financial networks behave, and as such, it is crucial to understand.

In particular, there is a range of natural questions that only arise if we study how the system develops in a step-by-step fashion. Can we reach every equilibrium state in a sequential manner? How does the ordering of default announcements influence the final outcome? Is there an optimal strategy of timing the announcements, either from a financial authority's or a single bank's perspective? How long can the sequential process last, and in particular, is it guaranteed to always stabilize eventually?

In this paper we analyze the development of financial systems in a sequential model, where banks update their situation one after another. We first study the \textit{reversible model}, which is a natural sequential setting in such networks. We analyze this model from three main perspectives:

\begin{itemize}[itemsep=0pt, parsep=5pt, topsep=5pt]
\item \textit{Stabilization time}: We show that a system can easily keep running infinitely in this model. Moreover, the time of stabilization heavily depends on the ordering of default announcements. We also present a more complex system that does stabilize eventually, but only after exponentially many steps.
\item \textit{Globally best solution}: We show that finding the ordering which results in the smallest (or largest) number of defaulting banks in the final state is NP-hard.
\item \textit{Defaulting strategies}: We study the best defaulting strategy of a single bank, and show that surprisingly, a bank may achieve the best outcome by announcing its default as early as possible. We also prove that in general, finding the best time to report a default is NP-hard.
\end{itemize}

Moreover, since the possibly infinite runtime is the most unrealistic aspect of this model, we analyze the reasons behind this phenomenon, and we discuss how it can be avoided in our sequential model.

\begin{itemize}[itemsep=0pt, parsep=5pt, topsep=5pt]
\item \textit{Monotone sequential model}: We show that with two minor changes to the setting (a more sophisticated update rule and a slightly different handling of defaulting banks), we can develop a monotone model variant where the recovery rate of banks can only decrease, and the system is always guaranteed to stabilize after quadratically many steps. We also compare this setting to the reversible model in terms of defaulting strategies.
\end{itemize}

\section{Related Work}

The network-based analysis of financial systems has been rapidly gaining attention in the last decade. Most studies are based on the early financial network model of Eisenberg and Noe \cite{model1}, which only assumes simple debt contracts between the banks. The propagation of shocks has been analyzed in many variants of this base model over the last decade \cite{prop2, prop3, prop4, prop5, prop6, coveredCDS}; in particular, the model has been extended by default costs \cite{veraart}, cross-ownership relations \cite{cross1, cross2} or game-theoretic aspects \cite{gametheo}.

However, the common ground in these model variants is that they only describe \textit{long positions} between banks: a better outcome for one bank always means a better (or the same) outcome for other banks. This already allows us to capture how the default of a single bank can cause a ripple effect in the system, but it also ensures that there is always an equilibrium which is simultaneously best for all banks \cite{model1, veraart}. As such, long positions cannot represent e.g. the opposing interests of banks in real-world systems. In particular, banks in practice often have \textit{short positions} on each other when a worse situation for one bank is more favorable to another bank, mostly due to various kinds of financial derivatives.

The recent work of Schuldenzucker et. al. \cite{base1} presents a more refined model where the network also contains credit default swaps (CDSs) besides regular debt contracts. CDSs are financial derivatives that essentially allow banks to bet on the default of another bank in the system; they have played a dominant role in the financial crisis of 2008 \cite{CDS3}, and have been thoroughly studied in the financial literature \cite{CDS1, CDS2}. While CDSs are still a rather simple kind of derivative, they already allow us to model short positions in the network; as such, their introduction to the system leads to remarkably richer behavior. In our paper, we also assume these two kinds of contracts in the network.

The work of \cite{base1, base2} discusses various properties of this new model: they show that systems may have multiple solutions (equilibrium states) in this model, and finding a solution is PPAD-complete. They also show that with default costs, these systems might not have a solution at all, and deciding whether a solution exists becomes NP-hard. The work of \cite{ec} studies a range of objective functions for selecting the best solution in this model, showing that the best equilibrium is not efficiently approximable to a $n^c$ factor for any $c<\frac{1}{4}$. The work of \cite{icalp} analyzes the model from a game-theoretical perspective, discussing how the removal or modification of contracts can lead to more favorable equilibria for the acting banks, and showing that such operations can lead to game-theoretical dilemmas.

However, all these results only analyze the model in terms of equilibrium states. This is indeed important when the market is hit by a large shock, and a central authority has to analyze the whole system, identify its equilibria, and possibly select one of them to artificially implement. However, apart from these rare occasions, the network mostly evolves sequentially, with banks announcing defaults in a step-by-step manner.
For an understanding of real-world networks, it is also essential to study this gradually developing behavior of the process besides the equilibrial outcomes.

Sequential models of financial networks have already been studied in several papers; however, most of them consider some variant of the debt-only model with long positions \cite{seqclear1, prop1}. The paper of \cite{base1} notes that sequential clearing in their model would be dependent on the order of defaults, but does not investigate this direction any further. At the other end of the scale, the work of \cite{banerjee} introduces a very general sequential model (where payment obligations can be a function of all banks and all previous time steps), with a specific focus on expressing concrete real-world examples in this setting. As such, to our knowledge, there is no survey that considers a simple network model with both long and short positions, and analyzes the step-by-step development of financial systems in this model.

Finally, we point out that the clearing problem indeed has a high relevance in practice, e.g. when financial authorities conduct stress tests to analyze the sensitivity of real-world networks. One concrete example for a study of this problem is the European Central Bank's stress test framework \cite{ECB}.

\section{Model Definition} \label{sec:finance}

\subsection{Banks and contracts}

Our financial system model consists of a set of \textit{banks} (or \textit{nodes}) $B$. We denote individual banks by $u$, $v$ or $w$, and the number of banks by $n=|B|$. Banks are connected by two kinds of contracts that both describe a specific payment obligation from a debtor bank $u$ to a creditor bank $v$. The amount of payment obligation is called the \textit{weight} of the contract.

The simpler kind of connection is a \textit{simple debt} contract, which obliges the debtor $u$ to pay a specific amount $\delta$ to the creditor $v$. This liability is unconditional, i.e. $u$ owes this amount to $v$ in any case.

Besides debts, banks can also enter into \textit{conditional debt contracts} where the payment obligation depends on some external event in the system. One of the most frequent forms of such a conditional debt is a credit default swap (\textit{CDS}), which obliges $u$ to pay a specific amount to $v$ in case a specific third bank $w$ (the \textit{reference entity}) is in default. More specifically, if $w$ can only fulfill a $r_w$ portion of its payment obligations (known as the recovery rate of $w$), then a CDS of weight $\delta$ implies a payment obligation of $\delta \cdot (1-r_w)$ from $u$ to $v$. For simplicity, we assume that all conditional debt contracts are CDSs.

In practice, CDS contracts can, for example, be used by a bank as an insurance policy against the default of its debtors. If $v$ suspects that its debtor $w$ might not be able to fulfill its payment obligation, then $v$ can enter into a CDS contract (as creditor) in reference to $w$; if $w$ goes into default and is indeed unable to pay, then $v$ receives some payment on this CDS instead. However, banks may also enter into CDSs for other reasons, e.g. speculative bets about future developments in the market. As a sanity assumption, we assume that no bank can enter into a contract with itself or in reference to itself.

Besides the contracts between banks, a financial system is described by the amount of funds (in financial terms: \textit{external assets}) owned by each bank, denoted by $e_v$ for a specific bank $v$. The external assets and the incoming payments describe the total amount of assets available to $v$, while the outgoing contracts describe the total amount of payment obligations of $v$. If $v$ is not able to fulfill all these obligations from its assets, then we say that $v$ is \textit{in default}. If $v$ is in default, then the fraction of liabilities that $v$ is able to pay is the \textit{recovery rate} of $v$, denoted by $r_v$. Note that $r_v \in [0,1]$, and $v$ is in default if $r_v<1$. We represent the recovery rates of all banks in a recovery rate vector $r \in [0,1]^B$.

For an example, consider the financial system in Figure \ref{fig:example1} with $3$ banks. The banks have external assets of $e_u=2$, $e_v=1$ and $e_w=0$. Bank $u$ has a debt of weight $2$ towards both $v$ and $w$, and there is a CDS of weight $2$ from $w$ to $v$, with $u$ as the reference entity.
In this network, $u$ has a total payment obligation of $4$, but only has assets of $2$, so $u$ is in default, with a recovery rate of $r_u=\frac{2}{4}=\frac{1}{2}$. Bank $u$ must use its funds of $2$ to pay $1$ unit of money to both $w$ and $v$, proportionally to its obligations. 
Since $r_u=\frac{1}{2}$, the CDS from $w$ to $v$ will induce a payment obligation of $2 \cdot (1-r_u)=1$. The payment of $1$ coming from $u$ allows $w$ to fulfill this obligation to $v$, thus narrowly avoiding default (hence $r_w=1$). Finally, $v$ receives $1$ unit from both $u$ and $w$, has funds of $1$ itself, and no payment obligations, so it has a positive equity of $3$, and $r_v=1$.

For convenience, we will use a simplified version of this notation in our figures: we only show the weight $\delta$ of a contract when $\delta \neq 1$, and we only show the external assets of $v$ explicitly if $e_v \neq 0$. We also write $e_v=\infty$ to conveniently indicate that $v$ can pay its liabilities in any case.

We also note that many of our constructions in the paper contain banks that have the exact same amount of assets and liabilities, like $w$ in this example. This is a somewhat artificial `edge case' that still ensures $r_w=1$. However, this is only for the sake of simplicity; we could avoid these edge cases by providing more assets to the banks in question.

Finally, we point out that contracts in a real-world financial system are often results of an earlier transaction between the banks, i.e. the creditor $v$ previously offering a loan to the debtor $u$. We assume that such earlier payments are implicitly represented in $e_u$, and as such, the external assets and the contracts are together sufficient to describe the current state of the system.

\begin{figure}
\centering
\hspace{0.08\textwidth}
\begin{subfigure}[b]{0.24\textwidth}

\begin{tikzpicture}

	\draw[very thick, blue, arrows=-latex] (0pt,0pt) -- (72pt,0pt);
	\draw[very thick, blue, arrows=-latex] (0pt,0pt) -- (34pt,44pt);
	\draw[very thick, brown, arrows=-latex] (40pt,50pt) -- (74pt,6pt);
	
	\node[anchor=center] at (40pt,6pt) {\small $2$};
	\node[anchor=center] at (14pt,27pt) {\small $2$};
	\node[anchor=center, rotate=-51.5] at (66pt,28pt) {\footnotesize $2\! \cdot\! (1\!-\!r{_{\!}}$\small$_v$\footnotesize$)$};
	
	\draw[black, fill=white] (0pt,0pt) circle (8pt);
	\draw[black, fill=white] (80pt,0pt) circle (8pt);
	\draw[black, fill=white] (40pt,50pt) circle (8pt);
	
	\node[anchor=center] at (0pt,0pt) {\large $u$};
	\node[anchor=center] at (80pt,0pt) {\large $v$};
	\node[anchor=center] at (40pt,50pt) {\large $w$};
	
	\draw [fill=white] (4.5pt,-2pt) rectangle (10.5pt,-11pt);
	\node[anchor=center] at (7.5pt,-6.5pt) {\footnotesize $2$};
	\draw [fill=white] (44.5pt,48pt) rectangle (50.5pt,39pt);
	\node[anchor=center] at (47.5pt,43.5pt) {\footnotesize $0$};
	\draw [fill=white] (84.5pt,-2pt) rectangle (90.5pt,-11pt);
	\node[anchor=center] at (87.5pt,-6.5pt) {\footnotesize $1$};
	
\end{tikzpicture}
	\vspace{-18pt}
	\captionsetup{justification=centering}
	\caption{}
	\vspace{-3pt}
	\label{fig:example1}
\end{subfigure}
\hspace{0.22\textwidth}
\begin{subfigure}[b]{0.24\textwidth}

\begin{tikzpicture}

	\draw[very thick, brown, arrows=-latex] (0pt,0pt) -- (72pt,0pt);
	\draw[very thick, blue, arrows=-latex] (80pt,0pt) -- (46pt,44pt);
	\draw[very thick, blue, arrows=-latex] (40pt,50pt) -- (6pt,6pt);
	
	\node[anchor=center] at (14pt,27pt) {\small $2$};
	\node[anchor=center] at (66pt,27pt) {\small $3$};
	\node[anchor=center] at (40pt,7pt) {\small $1\!-\!r_w$};
	
	\draw[black, fill=white] (0pt,0pt) circle (8pt);
	\draw[black, fill=white] (80pt,0pt) circle (8pt);
	\draw[black, fill=white] (40pt,50pt) circle (8pt);
	
	\node[anchor=center] at (0pt,0pt) {\large $u$};
	\node[anchor=center] at (80pt,0pt) {\large $v$};
	\node[anchor=center] at (40pt,50pt) {\large $w$};
	
	\draw [fill=white] (3pt,-3pt) rectangle (12pt,-10pt);
	\node[anchor=center] at (7.5pt,-6.5pt) {\scriptsize $\infty$};
	\draw [fill=white] (84.5pt,-2pt) rectangle (90.5pt,-11pt);
	\node[anchor=center] at (87.5pt,-6.5pt) {\footnotesize $1$};
	
\end{tikzpicture}
	\vspace{-18pt}
	\captionsetup{justification=centering}
	\caption{}
	\vspace{-3pt}
	\label{fig:example2}
\end{subfigure}
\hspace{0.1\textwidth}
\caption{Two example systems, with external assets shown in rectangles besides the banks. Simple debts are denoted by blue arrows from debtor to creditor, while CDSs are denoted by light brown arrows from debtor to creditor, with the payment obligation shown beside the arrow.}
\end{figure}
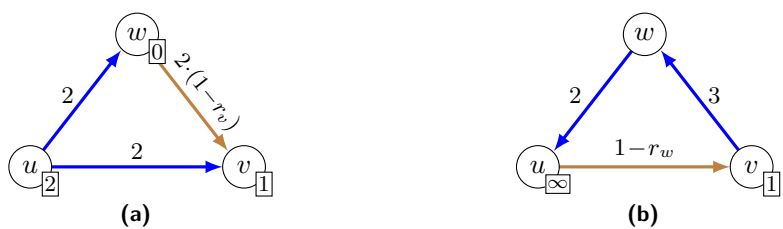

\subsection{Assets, liabilities and equilibria}
\label{sec:formalmodel}

We now formally define the liabilities and assets of banks in our systems. 
Note that due to the conditional debts, the payment obligations in a network are always a function of the recovery rate vector $r$.

Assuming a specific vector $r$, the liability $l_{u,v}$ of a bank $u$ towards a bank $v$ is defined as the sum of payment obligation from $u$ to $v$ on all contracts, i.e.
\[ l_{u,v}(r) = \delta_{u,v} +\sum_{w \in V} \delta_{u,v}^w \cdot (1-r_w), \]
where $\delta_{u,v}$ is the weight of the simple debt contract from $u$ to $v$ (if this contract exists, and $0$ otherwise), and $\delta_{u,v}^w$ is the weight of the CDS from $u$ to $v$ in reference to $w$ (if it exists, and $0$ otherwise). The total \textit{liability} of $u$ is simply the sum of liabilities to all other banks: $l_u(r) = \sum_{v \in V} l_{u,v}(r)$.

However, the actual \textit{payment} $p_{u,v}$ from $u$ to $v$ can be less than $l_{u,v}$ if $u$ is in default. If $u$ is in default, then it has to spend all of its assets to make payments to creditors. Most financial system models assume that in this case, $u$ has to follow the \textit{principle of proportionality}, i.e. it has to make payments proportionally to the corresponding liabilities. This means that if $u$ can pay an $r_u$ portion of its total liabilities, and it has a liability of $l_{u,v}$ towards $v$, then the payment from $u$ to $v$ is $p_{u,v}(r) = r_u \cdot l_{u,v}(r)$.

On the other hand, we can define the \textit{assets} of a bank $v$ as the sum of $v$'s external assets and its incoming payments in the network; that is,
\[ a_v(r) = e_v + \sum_{u \in V} p_{u,v}(r). \]
If $v$ is in default, then all these assets are used for $v$'s payment obligations; otherwise, $a_v-l_v$ of these assets remain at $v$. Note that while both $a_v(r)$ and $l_v(r)$ are formally a function of $r$, we often simplify this notation to $a_v$ and $l_v$ when the recovery rate is clear from the context.

Recall that the recovery rate of $v$ indicates the portion of payment obligations that $v$ is able to fulfill. As such, a valid choice of $r_v$ requires $r_v=1$ if we have $a_v \geq l_v$, and $r_v=\frac{a_v}{l_v}$ if $a_v < l_v$. For simplicity, let us introduce a separate function $R$ to denote this dependence on $a_v$ and $l_v$; that is, we define the function $R\,:\,[0,\infty) \times [0,\infty) \, \rightarrow \, [0,1]$ as
\[ R(a, l)=
\begin{cases}
    \: \, 1, & \text{if } a \geq l \\
		\: \, \frac{a}{l}, & \text{otherwise.}
\end{cases}
 \]

We say that a vector $r \in [0,1]^B$ is an \textit{equilibrium} (or a \textit{clearing vector}) of the system if for each bank $v \in B$, we have $r_v=R(_{\,} a_v(r), _{\,} l_v(r) _{\,} )$; that is, if the recovery rate vector is consistent with the assets and liabilities it generates in the network. Previous work has mostly focused on the analysis of different equilibrium states. Recall that while it is mostly straightforward to find the equilibrium states in our example constructions, the problem is PPAD-hard in general \cite{base2}.

We have already seen a simple example equilibrium in Figure \ref{fig:example1}; for another example that is slightly more challenging to compute, let us consider Figure \ref{fig:example2}. Here bank $u$ is again always able to pay its liabilities, so $r_u=1$ in any case. Furthermore, neither $r_v=1$ nor $r_w=1$ can provide an equilibrium in this network, so both $v$ and $w$ must be in default in any solution. Thus any equilibrium must have 
\[ r_v=\frac{a_v}{l_v}=\frac{1+1-r_w}{3} \qquad \text{and} \qquad r_w=\frac{a_w}{l_w}=\frac{3 \cdot r_v}{2} \, . \]
This implies that the only equilibrium is $r_v=\frac{4}{9}$, $r_w=\frac{2}{3}$.

\subsection{Sequential models of defaulting}

We have defined the equilibria of the system as the states $r$ that would fulfill the payment criteria if every bank were to simultaneously update its recovery rate to $r$. However, in practice, the announcement of defaults usually happens in a sequential manner, due to different sources of delay in the system: even if it is clear from $a_v$ and $l_v$ that a bank $v$ is only able to fulfill a specific $r_v$ portion of its liabilities, this might not be immediately known to the creditors of $v$ (due to incomplete information), or the legal framework may first allow $v$ to try to obtain further funds before officially having to announce its default. As such, the officially announced recovery rate $r_v$ might not always equal $R(a_v, l_v)$, and $v$ has to explicitly announce the changes in $r_v$ in order to make other banks aware of this situation.

Hence in our sequential model, each step of the process will consist of a single bank announcing an \textit{update} to its recovery rate. That is, given the assets $a_v$ and liabilities $l_v$ currently available to $v$, if the official recovery rate $r_v$ does not equal $R(a_v, l_v)$, then bank $v$ can (and eventually has to) announce a new official recovery rate of $r_v:=R(a_v, l_v)$. Since this affects both the payments received by the debtors of $v$ and the payment obligations on CDSs in reference to $v$, it can have various effects on the system, providing new assets and liabilities to some banks; as a result, these banks may also end up with a higher or lower asset/liability balance than their currently announced recovery rate, and thus they will also have to execute a new update at some point.

More formally, we consider discrete time steps $t=0,1,2,...$ . Each step consists of a single bank $v$ announcing an update to $r_v$. That is, if $v$ has assets $a_v\,\!\!^{(t-1)}$ and liabilities $l_v\,\!\!^{(t-1)}$, but a recovery rate of $r_v\,\!\!^{(t-1)} \neq R \left( a_v\,\!\!^{(t-1)}, l_v\,\!\!^{(t-1)} \right)$ at time $t-1$, then we say that $v$ is \textit{updatable} at time $t-1$. In each time step $t$, we select a bank $v$ that is updatable at time $t-1$, and define the state of the system at time $t$ by (i) setting $r_v\,\!\!^{(t)} = R \left( a_v\,\!\!^{(t-1)}, l_v\,\!\!^{(t-1)} \right)$ for the bank $v$ that executes the update, (ii) setting $r_u\,\!\!^{(t)} = r_u\,\!\!^{(t-1)}$ for every other bank $u \neq v$, and (iii)
calculating $a_u\,\!\!^{(t)}$ and $l_u\,\!\!^{(t)}$ for all $u \in B$ based on this new vector $r\,\!\!^{(t)}$.

We assume that initially, each bank $v$ has $r_v\,\!\!^{(0)}=1$, and we compute $a_v\,\!\!^{(0)}$ and $l_v\,\!\!^{(0)}$ accordingly. We say that the sequential process \textit{stabilizes} in round $t$ if there is no updatable bank in round $t$.

\section{Basic Properties} \label{sec:basics}

We begin by discussing some fundamental properties of this sequential setting.

\subsection{Reversibility and infinite cycling}

One important property of the sequential model is that even if a bank $v$ goes into default, it can easily return from this default later. That is, future updates in the system might increase the payment obligation on an incoming CDS of $v$, thus increasing $a_v$ and possibly raising $\frac{a_v}{l_v}$ above 1 again. This is in line with real-world financial systems, where returning from a default is also often possible if a bank acquires new assets. Due to this property, we also refer to this setting as the \textit{reversible model}.

Note that in practice, defaulting banks are often given a limited amount of time to obtain new assets and thus reverse a default; however, our sequential setting does not define an explicit timing of defaults (only their order), so such rules are not straightforward to include in our model. Nonetheless, we point out that many of our example constructions also work if we assume that defaults are only reversible for a specific (constant) number of rounds.

Another important property is that in a cyclic network topology, our model can easily result in an infinite loop of updates. Consider the example in Figure \ref{fig:infinite}, where the default of $v$ indirectly provides new assets to $v$. Since $r_v=1$ initially, $u$ must first update to $r_u=0$, and as a result, $v$ must update to $r_v=0$. However, this leads to new liabilities in the network, providing assets to both $u$ and (indirectly) to $v$, so $u$ (and then $v$) must update its rate back to $r_u=r_v=1$. This returns the system to its initial state, where $u$ (and $v$) will continue by updating their recovery rates to $0$ again.

If we keep repeating these few steps, then $u$ and $v$ alternate between $r_u=r_v=0$ and $r_u=r_v=1$ endlessly. Note that the system does have an equilibrium in $r_u=r_v=\frac{1}{2}$; however, instead of converging to this state, the banks keep on periodically repeating the same few steps. The possibility of such behavior in a sequential setting has already been noted in \cite{base1} or \cite{banerjee} before. While this looping behavior is certainly undesired, it follows straightforwardly from the reversibility of defaults and the existence of cycles in the network topology. As such, these situation could also occur in real-world systems, requiring a financial authority to intervene and set the system artificially to its equilibrium.

\subsection{Dependence on the order of updates}

Another key property of the sequential model is that the final outcome becomes dependent on the ordering of updates, i.e. whether some banks announce their default earlier or later.

We show a simple example of this dependence on the \textit{branching gadget} of Figure \ref{fig:branch}, which has already been used as a building block in the works of \cite{base2} and \cite{ec}. In this system, neither of the two banks $u$ and $v$ have any assets initially, so they are unable to fulfill their obligations. However, if $u$ is the first one to report default (updating to a new recovery rate of $a_v\,\!\!^{(0)}\!/\,l_v\,\!\!^{(0)}=0$), then this provides $1$ unit of new assets to $v$, which means that $v$ does not default anymore; the system stabilizes with $r_u=0$, $r_v=1$. Similarly, if $v$ is the first one to execute an update, then this provides new assets to $u$, and the system stabilizes with $r_u=1$, $r_v=0$. Thus both banks are strongly motivated to delay their default announcement as long as possible, as this might allow them to avoid defaulting entirely.

We can also note that there are further equilibrium states where both $u$ and $v$ are in default, e.g. when $r_u=\frac{1}{2}$ and $r_v=\frac{1}{2}$; due to its symmetry, one might even argue that this is the `fair' equilibrium to implement. However, this equilibrium is not reachable in any way through sequential updates; the only possible endstates of the sequential model are $(r_u, r_v)=(0,1)$ and $(r_u, r_v)=(1,0)$ as described above.

This shows that even in terms of the final outcome, the sequential model can significantly differ from the static analysis of the system. This is not due to the presence of fractional recovery rates: we can also easily have equilibria with integer (i.e., 0 or 1) recovery rates that is not reachable in a sequential setting. In Figure \ref{fig:unreach}, bank $u$ is the only node who can execute an update, which immediately leads to the unique final state $r_u=0$, $r_v=r_w=1$. However, $r_u=1$ with $r_v=r_w=0$ also forms an equilibrium in this system, so this phenomenon is indeed a result of the sequential nature of our model.

\begin{figure}
\hspace{0.01\textwidth}
\centering
\minipage{0.27\textwidth}
\centering
	\vspace{18pt}

\begin{tikzpicture}
	
	\draw[very thick, brown, arrows=-latex] (0pt,0pt) -- (71pt,0pt);
	\draw[very thick, blue, arrows=-latex] (80pt,0pt) -- (46pt,34pt);
	\draw[very thick, blue, arrows=-latex] (40pt,40pt) -- (6pt,6pt);
	
	\node[anchor=north] at (40pt,0pt) {\small $1\!-\!r_{_{\!}}$\small$_v$};
	
	\draw[black, fill=white] (0pt,0pt) circle (8.5pt);
	\draw[black, fill=white] (80pt,0pt) circle (8.5pt);
	\draw[black, fill=white] (40pt,40pt) circle (8.5pt);
	
	\node[anchor=center] at (80pt,0pt) {\large $u$};
	\node[anchor=center] at (40pt,40pt) {\large $v$};
	
	\draw [fill=white] (3pt,-3pt) rectangle (12pt,-10pt);
	\node[anchor=center] at (7.5pt,-6.5pt) {\scriptsize $\infty$};
	
\end{tikzpicture}
	\vspace{6pt}
	\caption{Example system for an infinite loop in the reversible model.}
	\label{fig:infinite}
\endminipage\hfill
\hspace{0.03\textwidth}
\minipage{0.32\textwidth}
\centering
	\vspace{16pt}

\begin{tikzpicture}
	
	\draw[very thick, blue, arrows=-latex] (50pt,25pt) -- (93pt,3.5pt);
	\draw[very thick, blue, arrows=-latex] (50pt,-25pt) -- (93pt,-3.5pt);
	\draw[very thick, brown, arrows=-latex] (0pt,0pt) -- (42pt,23pt);
	\draw[very thick, brown, arrows=-latex] (0pt,0pt) -- (42pt,-23pt);
	
	\node[anchor=center, rotate=29] at (22pt,19pt) {\footnotesize $1\!-\!r_{_{\!}}$\small$_v$};
	\node[anchor=center, rotate=-26] at (23pt,-20pt) {\footnotesize $1\!-\!r_{_{\!}}$\small$_u$};
	
	\draw[black, fill=white] (0pt,0pt) circle (8.5pt);
	\draw[black, fill=white] (50pt,25pt) circle (8.5pt);
	\draw[black, fill=white] (50pt,-25pt) circle (8.5pt);
	\draw[black, fill=white] (100pt,0pt) circle (8.5pt);
	
	\node[anchor=center] at (50pt,25pt) {\large $u$};
	\node[anchor=center] at (50pt,-25pt) {\large $v$};
	
	\draw [fill=white] (3pt,-3pt) rectangle (12pt,-10pt);
	\node[anchor=center] at (7.5pt,-6.5pt) {\scriptsize $\infty$};
	
\end{tikzpicture}
	\vspace{5pt}
	\caption{Example system where the outcome depends on the order of announcements.}
	\label{fig:branch}
\endminipage\hfill
\hspace{0.03\textwidth}
\minipage{0.31\textwidth}
\centering

\begin{tikzpicture}
	
	\draw[very thick, blue, arrows=-latex] (50pt,35pt) -- (94pt,5pt);
	\draw[very thick, brown, arrows=-latex] (50pt,-35pt) -- (94pt,-5pt);
	\draw[very thick, brown, arrows=-latex] (50pt,0pt) -- (92pt,0pt);
	\draw[very thick, brown, arrows=-latex] (0pt,35pt) -- (42pt,35pt);
	
	\node[anchor=center] at (23pt,41pt) {\footnotesize $1\!-\!r_{_{\!}}$\small$_v$};
	\node[anchor=center, rotate=35] at (76pt,-25pt) {\footnotesize $1\!-\!r_{_{\!}}$\small$_v$};
	\node[anchor=center] at (73pt,6pt) {\footnotesize $1\!-\!r_{_{\!}}$\small$_w$};
	
	\draw[black, fill=white] (0pt,35pt) circle (8.5pt);
	\draw[black, fill=white] (50pt,35pt) circle (8.5pt);
	\draw[black, fill=white] (50pt,0pt) circle (8.5pt);
	\draw[black, fill=white] (50pt,-35pt) circle (8.5pt);
	\draw[black, fill=white] (100pt,0pt) circle (8.5pt);
	
	\node[anchor=center] at (50pt,35pt) {\large $u$};
	\node[anchor=center] at (50pt,0pt) {\large $v$};
	\node[anchor=center] at (50pt,-35pt) {\large $w$};
	
	\draw [fill=white] (3pt,32pt) rectangle (12pt,25pt);
	\node[anchor=center] at (7.5pt,28.5pt) {\scriptsize $\infty$};
	
\end{tikzpicture}
	\caption{Example of an equilibrium that is not reachable in the sequential model.}
	\label{fig:unreach}
\endminipage\hfill
\hspace{0.01\textwidth}
\end{figure}

\section{Results} \label{sec:reverse}

We now move on to a deeper analysis of the model. We mainly focus on the length and outcome of the sequential process, and how the ordering of updates affects these properties.

Since our proofs will require more complex constructions, we switch to a simpler notation in our figures: instead of directly showing the liability $\delta \cdot (1-r_w)$ on a CDS, we only label the CDS by the weight $\delta$ and the reference entity $w$, or simply by $w$ when $\delta=1$. Nonetheless, recall that each such CDS still denotes a liability of $\delta \cdot (1-r_w)$.

\subsection{Stabilization time}

One fundamental question is the number of rounds it takes until the sequential process stabilizes, i.e. until no node can execute an update anymore. We first analyze this aspect in detail.

We have already seen in Figure \ref{fig:infinite} that even in simple examples, it can easily happen that the system does not stabilize at all.

\begin{corollary}
There is a system which never stabilizes.
\end{corollary}

Furthermore, with the appropriate ordering of default announcements, we can also obtain any finite value as a stabilization time.

\begin{lemma} \label{lem:stop}
For any integer $k$, there exists a system and an ordering such that the system stabilizes after exactly $k$ steps.
\end{lemma}

\begin{proof}
Consider the system on Figure \ref{fig:stopatwill}. Similarly to Figure \ref{fig:infinite}, this system allows us to produce an arbitrarily long sequence by switching only $u$ and $v$ repeatedly. However, when $w$ announces a default, then both $u$ and $v$ gain enough assets to fulfill their obligations, so the system stabilizes after at most 2 more updates.

This allows us reach any magnitude of stabilization time, apart from a constant offset. We can then simply add $O(1)$ more independent defaulting nodes to reach the desired value $k$. \qedhere
\end{proof}

This already shows that stabilization time can heavily depend on the order of updates. A more extreme case of this is when the choice of the first update already decides between two very different outcomes for the system.

\begin{lemma} \label{lem:branchtime}
There is a system where depending on the first update, the system either stabilizes in $1$ step, or does not ever stabilize.
\end{lemma}

\begin{proof}

Figure \ref{fig:difftime} is obtained by combining the base ideas of Figures \ref{fig:infinite} and \ref{fig:branch}. In this network, either bank $w_1$ or $w_2$ must execute the first update.

If $w_2$ is the first to announce $r_{w_2}=0$, then $w_1$ receives a payment of 1, and the system immediately stabilizes; no other bank will make an update.

However, if we update $r_{w_1}=0$ first, then $w_2$ survives, but on the other hand, $u$ receives no assets at all. In this case, nodes $u$ and $v$ are in the same situation as in Figure \ref{fig:infinite}, and thus the upper part of the system will never stabilize. \qedhere
\end{proof}

Finally, infinite loops are not the only examples of long stabilization: it is also possible that the system does stabilize eventually, but for any ordering of updates, this only happens after exponentially many steps.

\begin{theorem} \label{th:binary_counter}
There is a system where for any possible ordering, the system eventually stabilizes, but only after $2^{\Omega(n)}$ steps.
\end{theorem}

\renewcommand{\proofname}{Proof sketch.}

\begin{proof}
This proof requires a significantly more complex construction than our previous statements. We only outline the main idea of the construction here, and we discuss the details in Appendix \ref{App:A}.

The first step of the proof is to build a \textit{stable bit gadget}, which represents a mutable binary variable. The gadget offers a simple interface to set the bit to $0$ or $1$ through external conditions, and otherwise maintains its current value until the next such operation is executed.

Besides this, we create gadgets that describe logical states of an abstract process, similarly to a finite automaton. We also encode conditional transitions between these state gadgets, i.e. ensure that the system can only enter a given logical state if some banks currently have a specific recovery rate. This allows us to describe a logical process where the next state of the system is always determined by the current state and the current value of some stable bit gadgets.

Using these tools, we can essentially design a binary counter on $k=\Omega(n)$ bits, with $k$ stable bits representing the bits of the counter. This counter will proceed to count from $0$ to $2^k-1$, and only stabilize after the counting has finished, resulting in a sequence of at least $2^k$ steps.

The most challenging task is to ensure that in every step of the process, there is only one possible update we can execute next: the appropriate next step of the counting procedure. To achieve this, we not only need to ensure that some banks become updatable at specific times, but we also have to force the banks to indeed execute these updates, by encoding them as requirements in the transition conditions of our logical states. This results in a heavily restricted construction where there is essentially only one valid ordering of updates: the one that corresponds to the step-by-step incrementation of the binary counter.
\end{proof}

\renewcommand{\proofname}{Proof.}

Note that another possible approach for measuring the stabilization time of our systems is to consider the number of \textit{defaulting steps}, i.e. to only count the steps when a bank $v$ updates from $r_v=1$ to $r_v<1$. One can check that our results on stabilization time also hold for this alternative metric.

Finally, as a theoretical curiosity, we point out that our binary variable and state machine gadgets in the proof of Theorem \ref{th:binary_counter} demonstrate that we can essentially use financial networks as a model of computation. We discuss the expressive power of this model in Appendix \ref{App:TM}.

\begin{theorem} \label{th:TM}
We can use financial networks to simulate any Turing-machine with a finite tape.
\end{theorem}

\begin{figure}
\centering
\minipage{0.27\textwidth}
\centering
	\vspace{2pt}
	\resizebox{0.96\textwidth}{!}{

\begin{tikzpicture}

	\draw[very thick, brown, arrows=-latex] (0pt,0pt) -- (15pt,6pt) -- (60pt,6pt) -- (72pt,1pt);
	\draw[very thick, brown, arrows=-latex] (0pt,0pt) -- (15pt,-6pt) -- (60pt,-6pt) -- (72pt,-1pt);
	\draw[very thick, blue, arrows=-latex] (80pt,0pt) -- (46pt,44pt);
	\draw[very thick, blue, arrows=-latex] (40pt,50pt) -- (6pt,6pt);
	\draw[very thick, blue, arrows=-latex] (10pt,-35pt) -- (62pt,-35pt);
	
	\node[anchor=south] at (40pt,6pt) {\small $v$};
	\node[anchor=north] at (40pt,-6pt) {\small $w$};
	
	\draw[black, fill=white] (0pt,0pt) circle (8pt);
	\draw[black, fill=white] (80pt,0pt) circle (8pt);
	\draw[black, fill=white] (40pt,50pt) circle (8pt);
	\draw[black, fill=white] (10pt,-35pt) circle (8pt);
	\draw[black, fill=white] (70pt,-35pt) circle (8pt);
	
	\node[anchor=center] at (80pt,0pt) {\large $u$};
	\node[anchor=center] at (40pt,50pt) {\large $v$};
	\node[anchor=center] at (10pt,-35pt) {\large $w$};
	
	\draw [fill=white] (2pt,-4pt) rectangle (11pt,-11pt);
	\node[anchor=center] at (6.5pt,-7.5pt) {\scriptsize $\infty$};
	
\end{tikzpicture}}
	\vspace{0pt}
	\caption{Example system for Lemma \ref{lem:stop}. Recall that a CDS labeled with $v$ still describes a payment obligation of $1-r_v$.}
	\label{fig:stopatwill}
\endminipage\hfill
\hspace{0.03\textwidth}
\minipage{0.29\textwidth}
\centering
	\vspace{-3pt}
	\resizebox{1\textwidth}{!}{

\begin{tikzpicture}
	
	\draw[very thick, blue, arrows=-latex] (50pt,-20pt) -- (92pt,-3pt);
	\draw[very thick, blue, arrows=-latex] (50pt,20pt) -- (50pt,52pt);
	\draw[very thick, blue, arrows=-latex] (50pt,60pt) -- (82pt,60pt);
	\draw[very thick, brown, arrows=-latex] (0pt,0pt) -- (0pt,60pt) -- (42pt,60pt);
	\draw[very thick, blue, arrows=-latex] (90pt,60pt) -- (99pt,8pt);	
	\draw[very thick, brown, arrows=-latex] (0pt,0pt) -- (42pt,19pt);
	\draw[very thick, brown, arrows=-latex] (0pt,0pt) -- (42pt,-19pt);
	
	\node[anchor=center] at (21pt,16pt) {\small $w_2$};
	\node[anchor=center] at (21pt,-16pt) {\small $w_1$};
	\node[anchor=center] at (22pt,66pt) {\small $v$};
	
	\draw[black, fill=white] (0pt,0pt) circle (8pt);
	\draw[black, fill=white] (50pt,20pt) circle (8pt);
	\draw[black, fill=white] (50pt,-20pt) circle (8pt);
	\draw[black, fill=white] (50pt,60pt) circle (8pt);
	\draw[black, fill=white] (90pt,60pt) circle (8pt);
	\draw[black, fill=white] (100pt,0pt) circle (8pt);
	
	\node[anchor=center] at (50.5pt,19.5pt) {\large $w_1$};
	\node[anchor=center] at (50.5pt,-20.5pt) {\large $w_2$};
	\node[anchor=center] at (50pt,60pt) {\large $u$};
	\node[anchor=center] at (90pt,60pt) {\large $v$};
	
	\draw [fill=white] (2pt,-4pt) rectangle (11pt,-11pt);
	\node[anchor=center] at (6.5pt,-7.5pt) {\scriptsize $\infty$};
	
\end{tikzpicture}}
	\vspace{-12pt}
	\caption{Example system for Lemma \ref{lem:branchtime}, where stabilization time depends on the choice of the first update.}
	\label{fig:difftime}
\endminipage\hfill
\hspace{0.03\textwidth}
\minipage{0.36\textwidth}
\centering
	\vspace{13pt}
	\resizebox{1\textwidth}{!}{

\begin{tikzpicture}
	
	\draw[very thick, blue, arrows=-latex] (50pt,25pt) -- (72pt,25pt);
	\draw[very thick, blue, arrows=-latex] (50pt,-25pt) -- (72pt,-25pt);
	\draw[very thick, blue, arrows=-latex] (80pt,-25pt) -- (102pt,-25pt);
	\draw[very thick, blue, arrows=-latex] (118pt,-25pt) -- (132pt,-25pt);
	\draw[very thick, brown, arrows=-latex] (0pt,0pt) -- (42pt,23pt);
	\draw[very thick, brown, arrows=-latex] (0pt,0pt) -- (42pt,-23pt);
	
	\node[anchor=center] at (21pt,18pt) {\small $v$};
	\node[anchor=center] at (21pt,-18pt) {\small $u$};
	
	\draw[black, fill=white] (0pt,0pt) circle (8pt);
	\draw[black, fill=white] (50pt,25pt) circle (8pt);
	\draw[black, fill=white] (50pt,-25pt) circle (8pt);
	\draw[black, fill=white] (80pt,25pt) circle (8pt);
	\draw[black, fill=white] (80pt,-25pt) circle (8pt);
	\node[anchor=center] at (110pt,-29pt) {\small $...$};
	\draw[black, fill=white] (140pt,-25pt) circle (8pt);
	
	\node[anchor=center] at (50pt,25pt) {\large $u$};
	\node[anchor=center] at (50pt,-25pt) {\large $v$};
	
	\draw [fill=white] (2pt,-4pt) rectangle (11pt,-11pt);
	\node[anchor=center] at (6.5pt,-7.5pt) {\scriptsize $\infty$};
	
	\draw[gray, thick] (73pt,-37pt) -- (73pt,-40pt) -- (147pt,-40pt) -- (147pt,-37pt);
	\draw[gray, thick] (110pt,-43pt) -- (110pt,-40pt);
	\node[anchor=center] at (110pt,-49pt) {\footnotesize $n-4$ banks};
	
\end{tikzpicture}}
	\vspace{-13pt}
	\caption{Example system for Lemma \ref{lem:branching_best}, i.e. where the number of defaults depends on the choice of the first update.}
	\label{fig:diffoutcome}
\endminipage\hfill
\end{figure}

\subsection{Outcome with the fewest defaults} \label{sec:reversible_best}

In case of a larger shock, a financial authority could also be interested in the final state of the system, and in particular, the number of banks that end up in default. This can again heavily depend on the order of updates; in fact, even a single decision in the ordering can be critical from this perspective.

\begin{lemma} \label{lem:branching_best}
Depending on the first update, the number of defaults can be either $O(1)$ or $n-O(1)$.
\end{lemma}

\begin{proof}
Consider the system on Figure \ref{fig:diffoutcome}. If $u$ is the first to report a default with $r_u=0$, then $v$ receives 1 unit of payment, and thus no other node defaults. On the other hand, if $v$ reports a default first, then $u$ survives, but all the nodes in the lower chain have no incoming assets, and thus they all have to report a default eventually. So based on the first update, the number of defaults is either $1$ or $n-3$. \qedhere
\end{proof}

Hence if the authority has some influence over the ordering of updates, e.g. by allowing more flexibility to some banks than to others, then it could dramatically reduce the number of banks that end up in default. Unfortunately, even if we have complete control over the ordering, it is still hard to find the best possible ordering (in terms of the number of defaults in the final outcome).

\begin{theorem} \label{th:best_ordering_hard}
It is NP-hard to find the number of defaulting nodes in the best possible ordering.
\end{theorem}

\begin{proof}
We reduce the question to the MAXSAT problem: given a boolean formula in conjunctive normal form, the goal of MAXSAT is to find the assignment of variables that satisfies the highest possible number of clauses \cite{maxsat}.

Assume we have a MAXSAT problem on $k$ variables $x_1, ..., x_k$, and $m$ clauses. Note that in our financial systems, the branching gadget of Figure \ref{fig:branch} is a natural candidate for representing a boolean variable, since in any sequence, exactly one of $u$ and $v$ will eventually default. We point out that this gadget has already been used for similar purposes before in \cite{base2} and \cite{ec}. 

Hence for each variable $x_i$, we create a separate branching gadget in our system, and consider node $u$ to represent the literal $x_i$, and node $v$ to represent the literal $\neg x_i$. That is, we will consider $x_i=\textsc{true}$ if $u$ defaults, while we consider $x_i=\textsc{false}$ if $v$ defaults.

Furthermore, for each clause of the input formula, we create the clause gadget shown in Figure \ref{fig:maxsat_best}, with the CDSs labeled by the banks representing the literals in the clause. For example, the gadget in the figure is obtained for the clause $(x_1 \vee x_3 \vee \neg x_4)$. If any of the banks $x_1$, $x_3$ or $\neg x_4$ default, then $v$ receives enough assets to pay its debt, whereas otherwise, $v$ must eventually default.

If we aim to avoid as many defaults as possible, then the reasonable ordering strategy is to first evaluate all the variable gadgets, and the clause gadgets only afterwards. In this case, each bank $v$ of a clause gadget survives if and only if there is a true literal in the corresponding clause. This way the number of defaulting nodes in the final state is always exactly $k$ in the variable gadgets, and at most $m-\textsc{opt}$ in the clause gadgets, where $\textsc{opt}$ denotes the maximal number of satisfiable clauses in our MAXSAT problem. Thus the minimal number of defaulting nodes in the system is altogether $k+m-\textsc{opt}$. Finding this value also allows us to determine $\textsc{opt}$, which completes our reduction.
\end{proof}

To analyze the effects of a shock, one might also be interested in the worst possible ordering; a similar reduction shows that this is also hard to find.

\begin{theorem} \label{th:worst_ordering_hard}
It is NP-hard to find the number of defaulting nodes in the worst possible ordering.
\end{theorem}

\begin{proof}
We can apply the same reduction from MAXSAT as before; we only need to slightly change the clause gadgets. Consider the clause gadget of Figure \ref{fig:maxsat_worst} for the example clause $(x_1 \vee x_3 \vee \neg x_4)$. To maximize the number of defaulting banks in this system, we can first evaluate the variables gadgets, which then allows us to produce an extra default for each clause that has a true literal. Thus the maximum number of defaulting nodes is $k+\textsc{opt}$, which completes our reduction.
\end{proof}

\begin{figure}
\hspace{0.01\textwidth}
\centering
\minipage{0.42\textwidth}
\centering

\begin{tikzpicture}

	\draw[very thick, brown, arrows=-latex] (0pt,0pt) -- (15pt,13pt) -- (60pt,13pt) -- (74pt,4pt);
	\draw[very thick, brown, arrows=-latex] (0pt,0pt) -- (15pt,-13pt) -- (60pt,-13pt) -- (74pt,-4pt);
	\draw[very thick, brown, arrows=-latex] (0pt,0pt) -- (72.5pt,0pt);
	\draw[very thick, blue, arrows=-latex] (80pt,0pt) -- (122.5pt,0pt);
	
	\node[anchor=south] at (40pt,10.5pt) {\footnotesize $x_1$};
	\node[anchor=south] at (40pt,-2.5pt) {\footnotesize $x_3$};
	\node[anchor=south] at (39pt,-15.5pt) {\footnotesize $\neg x_4$};
	
	\draw[black, fill=white] (0pt,0pt) circle (8pt);
	\draw[black, fill=white] (80pt,0pt) circle (8pt);
	\draw[black, fill=white] (130pt,0pt) circle (8pt);
	
	\node[anchor=center] at (80pt,0pt) {\large $v$};
	
	\draw [fill=white] (2pt,-4pt) rectangle (11pt,-11pt);
	\node[anchor=center] at (6.5pt,-7.5pt) {\scriptsize $\infty$};
	
\end{tikzpicture}
	\vspace{4pt}
	\caption{Clause gadget for the MAXSAT reduction in Theorem \ref{th:best_ordering_hard}.}
	\label{fig:maxsat_best}
\endminipage\hfill
\hspace{0.1\textwidth}
\minipage{0.42\textwidth}
\centering

\begin{tikzpicture}

	\draw[very thick, brown, arrows=-latex] (0pt,0pt) -- (15pt,13pt) -- (60pt,13pt) -- (74pt,4pt);
	\draw[very thick, brown, arrows=-latex] (0pt,0pt) -- (15pt,-13pt) -- (60pt,-13pt) -- (74pt,-4pt);
	\draw[very thick, brown, arrows=-latex] (0pt,0pt) -- (72.5pt,0pt);
	
	\node[anchor=south] at (40pt,10.5pt) {\footnotesize $x_1$};
	\node[anchor=south] at (40pt,-2.5pt) {\footnotesize $x_3$};
	\node[anchor=south] at (39pt,-15.5pt) {\footnotesize $\neg x_4$};
	
	\draw[black, fill=white] (0pt,0pt) circle (8pt);
	\draw[black, fill=white] (80pt,0pt) circle (8pt);
	
	\node[anchor=center] at (0pt,0pt) {\large $v$};
	
\end{tikzpicture}
	\vspace{4pt}
	\caption{Clause gadget for the MAXSAT reduction in Theorem \ref{th:worst_ordering_hard}.}
	\label{fig:maxsat_worst}
\endminipage\hfill
\hspace{0.01\textwidth}
\end{figure}

\subsection{Individual defaulting strategies} \label{sec:strat}

It is also natural to consider the effect of the ordering from the perspective of a single bank $v$. More specifically, is $v$ motivated to immediately report its own default? Can it achieve a better outcome for itself by carefully timing its updates?

Intuitively, one would expect that banks are motivated to report their default as late as possible, in hope of obtaining further assets in the meantime. This is indeed true in many cases. For example, in the branching gadget of Figure \ref{fig:branch}, $u$ and $v$ clearly have a short position in each other, and if either of them can wait long enough such that the other bank reports a default first, then it obtains new assets from the incoming CDS and thus manages to avoid a default entirely.

However, due to the complex interconnections in a network, it is in fact also possible that $v$ achieves a better outcome if it reports a default earlier; it is even possible that this is the only strategy which allows $v$ to avoid a default in the endstate of the system. We consider this one of our most surprising results.

\begin{theorem} \label{th:earlydef_rev}
There exists a system where a bank $v_1$ can only avoid a default in the final state of the system if $v_1$ is the first bank to report a default.
\end{theorem}

\begin{proof}

Consider the system in Figure \ref{fig:earlydef_rev}, where only $v_1$ or $v_2$ can report a default initially, since no other node has any liabilities.

Assume that $v_1$ is the first to report a default, updating to $r_{v_1}=0$. This influences the network in two ways: $v_2$ obtains assets of $1$, and $u_2$ now has a new liability of $1$ as a result.

Thus the next update can only be executed by $u_2$, resulting in $r_{u_2}=0$. On the one hand, this provides assets to $u_1$; on the other hand, it creates liabilities for $w_2$. As a result, the next update can only be executed by $w_2$.

When $w_2$ announces $r_{w_2}=0$, this results in more liabilities for the defaulting $u_2$, and more assets for $v_1$. These assets make $v_1$ the only updatable next node, bringing $v_1$ back from its default with $r_{v_1}=1$.

When $v_1$ announces $r_{v_1}=1$, then $u_2$ loses some of its liabilities, and $v_2$ loses its assets. This does not affect $u_2$, which remains at $r_{u_2}=0$ due to the default of $w_2$; however, $v_2$ now also has to report a default. The system finally stabilizes after $v_2$ updates to $r_{v_2}=0$: the assets/liabilities of $v_1$ and $u_1$ are affected, but neither of them has to make an update. Thus the final solution has $r_{v_1}=1$ and $r_{v_2}=0$.

On the other hand, if $v_2$ is the first to report default, then due to the symmetry of the system, the final outcome will have $r_{v_1}=0$ and $r_{v_2}=1$. Note that in both cases, after the first update is executed, the remaining steps are already determined, and no alternative ordering is possible. Hence the only way for $v_1$ to avoid a default in the final outcome is to be the first one to report a default. \qedhere
\end{proof}

We can also show that in general, it is NP-hard to find the best default-reporting strategy for a bank. This even holds if the behavior of the rest of the network is `predictable', i.e. if there is essentially only one ordering that the system can follow. This implies that any interpretation of this problem, e.g. optimizing a bank's best-case payoff or worst-case payoff, is also hard.

\begin{theorem} \label{th:findbesttime}
It is NP-hard to find the time of defaulting that provides the highest payoff to a specific bank in the final outcome.
\end{theorem}

\renewcommand{\proofname}{Proof sketch.}

\begin{proof}
The main idea of the proof is to combine the binary counter construction of Theorem \ref{th:binary_counter} with the MAXSAT reduction. That is, given a binary counter on $k=\Theta(n)$ bits, we add a new node $v$ to the system such that

\begin{enumerate}[(a)]
\setlength{\itemsep}{7pt}
\setlength{\parskip}{-5pt}
\item $v$ can choose to default anytime,
\item the default of $v$ terminates the counting process, stabilizing the counter in its current state,
\item $v$ then comes back from its default, and its assets in the final state are proportional to the amount of clauses satisfied in a SAT formula, where the value of the variables is derived from the finalized state of the bits in the counter.
\end{enumerate}
This means that the counter essentially enumerates all the possible value assignments of the variables, and the best defaulting strategy is obtained if counting is terminated at the assignment that satisfies the highest number of clauses. However, finding this assignment is NP-hard.

The details of the construction are discussed in Appendix \ref{App:B}.
\end{proof}

\renewcommand{\proofname}{Proof.}

\begin{figure}
\centering
\minipage{0.63\textwidth}
\centering

\begin{tikzpicture}
	
	\draw[very thick, blue, arrows=-latex] (70pt,30pt) -- (96pt,5.5pt);
	\draw[very thick, blue, arrows=-latex] (70pt,-30pt) -- (96pt,-5.5pt);
	\draw[very thick, brown, arrows=-latex] (10pt,30pt) -- (25pt,36pt) -- (50pt,36pt) -- (62pt,31pt);
	\draw[very thick, brown, arrows=-latex] (10pt,30pt) -- (25pt,24pt) -- (50pt,24pt) -- (62pt,29pt);
	\draw[very thick, brown, arrows=-latex] (10pt,-30pt) -- (25pt,-36pt) -- (50pt,-36pt) -- (62pt,-31pt);
	\draw[very thick, brown, arrows=-latex] (10pt,-30pt) -- (25pt,-24pt) -- (50pt,-24pt) -- (62pt,-29pt);

	\node[anchor=center] at (37.5pt,41pt) {\small $v_2$};
	\node[anchor=center] at (37.5pt,18pt) {\small $w_2$};
	\node[anchor=center] at (37.5pt,-18pt) {\small $v_1$};
	\node[anchor=center] at (37.5pt,-42pt) {\small $w_1$};
	
	\draw[black, fill=white] (10pt,30pt) circle (1.7ex);
	\draw[black, fill=white] (10pt,-30pt) circle (1.7ex);
	\draw[black, fill=white] (70pt,30pt) circle (1.7ex);
	\draw[black, fill=white] (70pt,-30pt) circle (1.7ex);
	\draw[black, fill=white] (100pt,0pt) circle (1.7ex);
	
	\node[anchor=center] at (70.5pt,29.5pt) {\large $v_1$};
	\node[anchor=center] at (70.5pt,-30.5pt) {\large $v_2$};
	
	\draw [fill=white] (12pt,25pt) rectangle (21pt,18pt);
	\node[anchor=center] at (16.5pt,21.5pt) {\scriptsize $\infty$};
	\draw [fill=white] (12pt,-35pt) rectangle (21pt,-42pt);
	\node[anchor=center] at (16.5pt,-38.5pt) {\scriptsize $\infty$};
	
	
	\draw[very thick, brown, arrows=-latex] (140pt,0pt) -- (173pt,17pt);
	\draw[very thick, brown, arrows=-latex] (140pt,-0pt) -- (173pt,-17pt);
	\draw[very thick, brown, arrows=-latex] (180pt,20pt) -- (195pt,26pt) -- (220pt,26pt) -- (232pt,21pt);
	\draw[very thick, brown, arrows=-latex] (180pt,20pt) -- (195pt,14pt) -- (220pt,14pt) -- (232pt,19pt);
	\draw[very thick, brown, arrows=-latex] (180pt,-20pt) -- (195pt,-26pt) -- (220pt,-26pt) -- (232pt,-21pt);
	\draw[very thick, brown, arrows=-latex] (180pt,-20pt) -- (195pt,-14pt) -- (220pt,-14pt) -- (232pt,-19pt);
	\draw[very thick, brown, arrows=-latex] (180pt,45pt) -- (220pt,45pt) -- (236pt,26pt);
	\draw[very thick, brown, arrows=-latex] (180pt,-45pt) -- (220pt,-45pt) -- (236pt,-26pt);
	
	\node[anchor=center] at (156pt,16pt) {\small $u_2$};
	\node[anchor=center] at (157pt,-16pt) {\small $u_1$};
	\node[anchor=center] at (207.5pt,31pt) {\small $v_2$};
	\node[anchor=center] at (207.5pt,8pt) {\small $w_1$};
	\node[anchor=center] at (207.5pt,-8pt) {\small $v_1$};
	\node[anchor=center] at (207.5pt,-32pt) {\small $w_2$};
	\node[anchor=center] at (211pt,50pt) {\small $u_1$};
	\node[anchor=center] at (211pt,-50.5pt) {\small $u_2$};
	
	\draw[black, fill=white] (240pt,20pt) circle (8pt);
	\draw[black, fill=white] (240pt,-20pt) circle (8pt);
	\draw[black, fill=white] (180pt,20pt) circle (8pt);
	\draw[black, fill=white] (180pt,-20pt) circle (8pt);
	\draw[black, fill=white] (140pt,0pt) circle (8pt);
	\draw[black, fill=white] (180pt,45pt) circle (8pt);
	\draw[black, fill=white] (180pt,-45pt) circle (8pt);
	
	\node[anchor=center] at (180.5pt,19.5pt) {\large $u_1$};
	\node[anchor=center] at (180.5pt,-20.5pt) {\large $u_2$};
	\node[anchor=center] at (180.5pt,44.5pt) {\large $w_1$};
	\node[anchor=center] at (180.5pt,-45.5pt) {\large $w_2$};
	
	\draw [fill=white] (142pt,-5pt) rectangle (151pt,-12pt);
	\node[anchor=center] at (146.5pt,-8.5pt) {\scriptsize $\infty$};
	
\end{tikzpicture}
	\caption{Example where early defaulting is the best strategy, with multiple source and sink nodes for a cleaner topology.}
	\label{fig:earlydef_rev}
\endminipage\hfill
\hspace{0.03\textwidth}
\minipage{0.31\textwidth}
\centering

\begin{tikzpicture}
	
	\draw[very thick, blue, arrows=-latex] (0pt,40pt) -- (71pt,40pt);
	\draw[very thick, blue, arrows=-latex] (80pt,40pt) -- (46pt,6pt);
	\draw[very thick, blue, arrows=-latex] (40pt,0pt) -- (6pt,34pt);
	\draw[very thick, blue, arrows=-latex] (0pt,40pt) -- (34pt,74pt);
	
	\draw[black, fill=white] (0pt,40pt) circle (8.5pt);
	\draw[black, fill=white] (40pt,0pt) circle (8.5pt);
	\draw[black, fill=white] (40pt,80pt) circle (8.5pt);
	\draw[black, fill=white] (80pt,40pt) circle (8.5pt);
	
	\node[anchor=center] at (0pt,40pt) {\large $u$};
	\node[anchor=center] at (40pt,0pt) {\large $w$};
	\node[anchor=center] at (80pt,40pt) {\large $v$};
	
\end{tikzpicture}
	\vspace{3pt}
	\caption{Infinite convergence to an equilibrium state.}
	\label{fig:convergence}
\endminipage\hfill
\end{figure}

\section{Achieving Stabilization} \label{sec:monoton}

While the reversible sequential model is realistic from many perspectives, the infinite looping property is clearly not reasonable in real-world systems. As such, it is natural to ask if there is a way to modify the model to avoid this situation, and instead ensure that every financial system stabilizes eventually.

In this section, we investigate the causes of this infinite behavior in the sequential model. We first show that we require more sophisticated update rules to avoid a specific kind of infinite behavior, namely when the system converges to an equilibrium. We then discuss liability freezing, a different (but in some sense also realistic) approach of handling defaulting banks in the network. Finally, we show that if we combine these two modifications, we can obtain a monotone sequential model where our systems always stabilize after polynomially many steps.

\subsection{More sophisticated update rules}

\subparagraph*{Convergence to an equilibrium.}

Since the addition of conditional debt contracts drastically increases the complexity of the model, it is a natural first assumption that such an infinite pattern can only arise if the system contains a CDS. However, this is not the case: we can also obtain a (slightly different kind of) infinite sequence in systems with only regular debts.

Consider the example system in Figure \ref{fig:convergence}. Since bank $u$ has $l_u=2$ and $a_u=1$ initially, it can begin by updating its recovery rate to $r_u=\frac{1}{2}$. As a result, $v$ and $w$ must also announce recovery rates of $r_v=r_w=\frac{1}{2}$. With $a_u=\frac{1}{2}$, bank $u$ now has to update to $r_u=\frac{1}{4}$, which then gives $r_v=r_w=\frac{1}{4}$. Each such round prompts another round of updates, slowly converging to $r_u=r_v=r_w=0$. While this is indeed the only equilibrium of the system, the process takes infinitely many steps to reach this state.

\subparagraph*{Explicit computation of equilibria.}

Such a convergence process can easily occur in any network with cycles; as real-world financial systems are also known to contain cycles \cite{cycles}, we can easily encounter such a situation in practice. In this case, it seems
that a financial authority (or the banks involved) have no other option than to explicitly compute this equilibrium, and set their recovery rates to the appropriate values.

Fortunately, it is known that in case of fixed liabilities in the network (i.e. only simple debts), this is computationally feasible: there always exists a single maximal solution that is simultaneously best for all banks, and this solution can be found in polynomial time \cite{veraart}, essentially by repeatedly solving a system of linear equations. Thus an authority could indeed find this solution, and banks could directly update to these recovery rates in order to skip the convergence steps.

This allows us to introduce the notion of \textit{smart updates}: after each updating step, we can consider the current liabilities in the network fixed, and we assume that the equilibrium of the system is computed under these liabilities (essentially reducing the convergence process to a single step). This equilibrium defines a \textit{tentative recovery rate} for each bank $v$, denoted by $\overline{r_v}$. In smart updates, we assume that whenever $v$ executes an update, it always updates to $r_v:=\overline{r_v}$.

In the example of Figure \ref{fig:convergence}, this means that the tentative recovery rates $\overline{r_u}=\overline{r_v}=\overline{r_w}=0$ are already computed initially, and thus any bank executing an update will immediately set its recovery rate to $0$. This way the process already stabilizes after each bank has executed one update. In general, we achieve stabilization in this setting when $r_v=\overline{r_v}$ for each bank $v$ in the network.

While the explicit computation of equilibria may seem artificial, in practice, defaulting banks are often subject to more thorough supervision by the authorities. As such, it is not so unrealistic that the situation of a defaulting bank $v$ is first analyzed by an authority, and this analysis determines the official recovery rate of $v$.

Also, recall that while equilibria are easy to find in debt-only networks, the introduction of CDSs changes this picture entirely. With CDSs, there can easily be multiple equilibria that are Pareto-optimal, and finding any of them is already a PPAD-hard problem \cite{base2}. Thus this explicit computation of $\overline{r_v}$ is only possible for a single step of the process, when we consider the current payment obligation on each CDS fixed. As defaults rarely happen simultaneously in practice, it can indeed be realistic to assume that we can analyze the current (fixed) liabilities in the network after each new update.

Finally, note that smart updating is not yet enough to avoid an infinite convergence. In the system shown in Figure \ref{fig:other}, $v$ can initially fulfill its obligations, while $u$ must update to $r_u=\frac{1}{2}$. This creates new liabilities of $2$ for $v$, leading to the tentative recovery rates $\overline{r_v}=\frac{2}{3}$ and thus $\overline{r_u}=\frac{1}{3}$ after this first step. If $u$ updates again (to $r_u=\frac{1}{3}$), then the liability on the CDS again increases, and thus the next computed equilibrium has an even lower $\overline{r_u}$.

Each step of this process provides new tentative recovery rates, obtained as $\overline{r_v}=\frac{2}{5-4 \cdot r_u}$ and $\overline{r_u}=\frac{\overline{r_v}}{2}=\frac{1}{5-4 \cdot r_u}$. This results in an infinite convergence to the equilibrium $r_v=\frac{1}{2}$, $r_u=\frac{1}{4}$. Note that we can observe this behavior regardless of whether $v$ ever updates its recovery rate to the new $\overline{r_v}$ value; the assets of $u$ are calculated independently of the recovery rate reported by $v$.

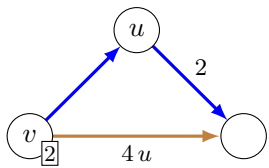
\begin{figure}
\centering
    \vspace{10pt}

\begin{tikzpicture}
	
	\draw[very thick, brown, arrows=-latex] (0pt,0pt) -- (71pt,0pt);
	\draw[very thick, blue, arrows=-latex] (40pt,40pt) -- (74pt,6pt);
	\draw[very thick, blue, arrows=-latex] (0pt,0pt) -- (34pt,34pt);
	
	\node[anchor=north] at (40pt,0pt) {\small $4 \, u$};
	\node[anchor=center] at (64pt,26pt) {\small $2$};
	
	\draw[black, fill=white] (0pt,0pt) circle (8.5pt);
	\draw[black, fill=white] (80pt,0pt) circle (8.5pt);
	\draw[black, fill=white] (40pt,40pt) circle (8.5pt);
	
	\node[anchor=center] at (0pt,0pt) {\large $v$};
	\node[anchor=center] at (40pt,40pt) {\large $u$};
	
	\draw [fill=white] (4.5pt,-2pt) rectangle (10.5pt,-11pt);
	\node[anchor=center] at (7.5pt,-6.5pt) {\footnotesize $2$};
	
\end{tikzpicture}
	\vspace{4pt}
	\caption{Example of infinite convergence in a network, even in case of smart updates. Recall that the label $4u$ on the CDS describes a payment obligation of $4 \cdot (1-r_u)$.}
	\label{fig:other}
\end{figure}

\subparagraph*{Optimistic updates.}

Another natural variant of smart updates is the \textit{optimistic update} rule. To avoid the convergence phenomenon of Figure \ref{fig:convergence}, this setting also assumes that the system is analyzed by an authority after each update. However, defaulting and non-defaulting nodes are now handled in a different manner in this analysis. More specifically, if a bank $v$ is not in default (it has $r_v=1$ currently), then it is given the benefit of a doubt: we assume that it can fulfill its obligations, regardless of how many assets it currently has. On the other hand, banks in default are handled the same way as in case of smart updates.

This distinction can indeed be realistic: if $v$ is non-defaulting, then $a_v$ might not even be known to other banks, so the creditors of $v$ have no better option than to assume that they will receive all payments from $v$. On the other hand, the assets of defaulting banks are under more thorough scrutiny in most legal frameworks.

Formally, optimistic update means that after each step of the process, we use a modified version of the liability network to compute the equilibrium. Whenever there is a contract of current weight $\delta$ from $u$ to $v$ with $r_u=1$, then we remove this contract from the network, and instead (i) we add a new debt of weight $\delta$ from $u$ to an artificial sink node $s$, ensuring that $u$ still has this liability, and (ii) we increase the value of $e_v$ by $\delta$, ensuring that $v$ always has these assets. In contrast, if $r_u<1$, we do not execute any changes on the outgoing contracts. This modified network ensures that until a bank reports a default, its lack of assets does not affect its creditors. We then use the same algorithm of \cite{veraart} to find the equilibrium in this modified system, and set the next tentative recovery rates accordingly.

Revisiting the system in Figure \ref{fig:other}, we see that bank $u$ can again first update to $r_u=\frac{1}{2}$, which results in $\overline{r_v}=\frac{2}{3}$. However, with optimistic updates, $u$ cannot make an update again: until $v$ adjusts its recovery rate to this new value, the tentative recovery rate of $u$ remains $\frac{1}{2}$, since we still expect to get the entire payment from the non-defaulting $v$. Note, however, that optimistic updating still does not prevent an infinite convergence in this system if, for example, $u$ and $v$ keep on updating alternatingly.

\subsection{Liability freezing}

We have seen that neither smart nor optimistic updating prevents an infinite sequential process in itself. For this, we also need to change another aspect of our model, namely how the contracts of $v$ are handled once $v$ goes into default.

Debts are rather simple from this perspective: they describe a previously established payment obligation in the network, so there is no incentive to change them if $v$ defaults.

CDSs, however, pose a more complicated question, since they describe payment obligations that are dynamically changing. So far, we assumed that even after $v$ defaults, the payment obligations on its CDSs keep changing as the reference entities are updated. Another possible approach is to assume \textit{liability freezing}: whenever $v$ goes into default, the liabilities on any incoming or outgoing CDS are fixed at the current value for the rest of the process. That is, a CDS with weight $\delta$ and reference entity $w$ at time $t$ is essentially converted into a simple debt contract with weight $\delta \cdot (1-r_w\,\!\!^{(t)})$, and this weight does not change in the future, even if $r_w$ is updated.

This can be realistic when there is a larger time difference between subsequent defaults: by the time the next default happens, the previous bank has already completed the first phase of the insolvency process, and its incoming/outgoing payments have been established and fixed. Indirectly, such a framework suggests that if $v$ defaults, then it is expected to immediately `cash in' its incoming debts and fulfill its payment obligations, and not wait for a more favorable situation.

The main advantage of this approach is that if we combine liability freezing with optimistic updates, it provides a \textit{monotone sequential model} where recovery rates can only decrease throughout the process. Intuitively, when a bank $w$ makes an update, then CDSs in reference to $w$ could only provide more assets to a bank $v$ if we still have $r_v=1$, as otherwise the liability on the CDS is already fixed. However, if $r_v=1$, then the optimistic approach assumes anyway that $v$ can pay its liabilities, and thus the update has no effect on other banks in the system.

This monotonic property ensures that any system stabilizes eventually in this model; on the other hand, it also means that once a bank $v$ announces a default in this model, it has no possibility to reverse this default in the future, and its recovery rate can only get smaller with further updates.

By revisiting Figure \ref{fig:infinite}, we can observe that both liability freezing and optimistic updates are crucial ingredients to achieve this monotonicity. Without liability freezing, the system loops infinitely if $u$ and $v$ make updates in an alternating fashion, both with smart and with optimistic updates. On the other hand, if we combine liability freezing with smart updates, then $v$ can still alternate between $r_v=0$ and $r_v=1$ indefinitely; if $u$ never makes an update, then the liability on the CDS will never be fixed at a specific value.

\subsection{Stabilization in the monotone model} \label{sec:monprop}

We now discuss the main properties of the monotone model. We first show that the model indeed ensures an eventual stabilization for any ordering. The key observation for this is that the recovery rate of banks can never increase in this model.

\begin{theorem} \label{th:monotone}
The recovery rate of a bank can only decrease in the monotone model.
\end{theorem}

\begin{proof}
The main idea is to show that for any bank $v$, $\overline{r_v}$ can only increase if we still have $r_v=1$ currently. This shows that we can never have $\overline{r_v} > r_v$, and thus no update can increase $r_v$.

Assume that node $w$ updates $r_w$ in a specific step, and assume for contradiction that this is the first step that increases $\overline{r_v}$ for some bank $v$ with $r_v < 1$. This means that the current update is still a decrease of $r_w$, since we must have $\overline{r_w} < r_w$. The update of $r_w$ can have two kinds of effects on the system: it can change the liabilities on CDSs that are in reference to $w$, and it can result in a lower amount of assets for the creditors of $w$. We analyze these two effects separately.

Since the monotone model has liability freezing, the liability on a CDS from $u$ to $v$ (in reference to $w$) can only change if we currently still have $r_u=r_v=1$. Thus while this extra payment may increase $\overline{r_v}$, we will still have $\overline{r_v} \leq r_v$ afterwards. Since the model uses optimistic updates and $r_u=r_v=1$, both $u$ and $v$ only have debts towards the artificial sink $s$ in the input graph of the equilibrium algorithm (which computes the tentative recovery rates), so the changes to $\overline{r_u}$ and $\overline{r_v}$ do not affect the tentative recovery rate of any other node.

As for the creditors of $w$, we consider two cases. If this is not a defaulting step (we already had $r_w<1$ before the update), then updating $r_w$ does not change the liabilities in the input graph of the equilibrium algorithm (apart from the case of some non-defaulting nodes, as discussed above), so the tentative recovery rates will remain unchanged.

On the other hand, if this is a defaulting step, then the outgoing debts of $w$ will now be redirected from $s$ to the actual creditors of $w$. However, this operation can only result in less assets for a bank. More specifically, one can observe that any configuration of payments in this new graph is also a valid configuration of payments in the original graph before the redirection step. Hence if the $\overline{r_v}$ value of any bank $v$ increases with this step, then this contradicts the fact that the previous $\overline{r_v}$ was obtained from a maximal equilibrium of the system.
\end{proof}

\begin{theorem} \label{th:upper_quad}
The monotone model allows at most $n$ defaulting and $O(n^2)$ updating steps.
\end{theorem}

\begin{proof}
Since recovery rates are always decreasing, every bank can default at most once, thus the number of defaulting steps is at most $n$.

For the $O(n^2)$ upper bound, we show that there are at most $n$ updating steps between any two consecutive defaulting steps. This is rather straightforward: recall from the proof of Theorem \ref{th:monotone} that if bank $w$ executes a non-defaulting update, then this can only change the value of $\overline{r_v}$ for banks $v$ that are not in default. Thus for any bank $v$ in default, $\overline{r_v}$ can not change between two defaulting steps of the process. This means that any bank can execute at most $1$ updating step between two consecutive defaulting steps, limiting the number of steps in this period to $n$.
\end{proof}

Furthermore, we point out that this upper bound is asymptotically tight: we can easily construct a system and an ordering that indeed lasts for $\Omega(n^2)$ steps in the monotone model. This construction does not even require CDSs in the network; it only contains simple debt contracts.

\begin{lemma} \label{lem:longstab}
There is a system with an ordering that lasts for $\Omega(n)$ defaulting and $\Omega(n^2)$ updating steps.
\end{lemma}

\begin{proof}
Let $m$ be a parameter with $m=\Theta(n)$, and consider Figure \ref{fig:longstab}. All the banks $w_1$, ..., $w_m$ will eventually report a default in this system, so the number of defaulting steps is indeed $m=\Omega(n)$.

Let $w_1$, ..., $w_m$ report a default in this order throughout the process. After $w_i$ has reported a default, bank $v$ can always decrease its recovery rate to a new value of $r_v=\frac{m-i}{m}$. Finally, after each such update of $v$, assume that all the nodes $u_1$, ..., $u_m$ make an update step, also announcing a new recovery rate of $\frac{m-i}{m}$; they can indeed all do this due to the update executed by $v$. This ordering has $\Omega(m^2)=\Omega(n^2)$ updating steps altogether.
\end{proof}

\begin{figure}
\centering
\minipage{0.32\textwidth}
\centering

\begin{tikzpicture}
	
	\draw[very thick, blue, arrows=-latex] (50pt,0pt) -- (5pt,35pt);
	\draw[very thick, blue, arrows=-latex] (50pt,0pt) -- (32pt,33pt);
	\draw[very thick, blue] (50pt,0pt) -- (53pt,19pt);
	\draw[very thick, blue] (50pt,0pt) -- (61pt,19pt);
	\draw [white, fill=white] (50pt,18pt) rectangle (64pt,21pt);
	\draw[very thick, blue, arrows=-latex] (50pt,0pt) -- (95pt,35pt);
	
	\draw[very thick, blue, arrows=-latex] (0pt,40pt) -- (45pt,75pt);
	\draw[very thick, blue, arrows=-latex] (30pt,40pt) -- (48pt,73pt);
	\draw[very thick, blue, arrows=-latex] (53pt,58pt) -- (51pt,72pt);
	\draw[very thick, blue, arrows=-latex] (61pt,58pt) -- (53pt,73pt);
	\draw [white, fill=white] (50pt,59pt) rectangle (64pt,56pt);
	\draw[very thick, blue, arrows=-latex] (100pt,40pt) -- (55pt,75pt);
	
	\draw[very thick, blue, arrows=-latex] (0pt,-40pt) -- (45pt,-5pt);
	\draw[very thick, blue, arrows=-latex] (30pt,-40pt) -- (48pt,-7pt);
	\draw[very thick, blue, arrows=-latex] (53pt,-22pt) -- (51pt,-8pt);
	\draw[very thick, blue, arrows=-latex] (61pt,-22pt) -- (53pt,-7pt);
	\draw [white, fill=white] (50pt,-21pt) rectangle (64pt,-24pt);
	\draw[very thick, blue, arrows=-latex] (100pt,-40pt) -- (55pt,-5pt);
	
	\draw[black, fill=white] (50pt,0pt) circle (8.5pt);
	\draw[black, fill=white] (0pt,40pt) circle (8.5pt);
	\draw[black, fill=white] (30pt,40pt) circle (8.5pt);
	\node[anchor=center] at (65pt,36pt) {\Large ...};
	\draw[black, fill=white] (100pt,40pt) circle (8.5pt);
	
	\draw[black, fill=white] (50pt,80pt) circle (8.5pt);
	
	\draw[black, fill=white] (0pt,-40pt) circle (8.5pt);
	\draw[black, fill=white] (30pt,-40pt) circle (8.5pt);
	\node[anchor=center] at (65pt,-44pt) {\Large ...};
	\draw[black, fill=white] (100pt,-40pt) circle (8.5pt);
	
	\node[anchor=center] at (50pt,0pt) {\normalsize $v$};
	\node[anchor=center] at (0pt,39.5pt) {\normalsize $u_1$};
	\node[anchor=center] at (30.5pt,39.5pt) {\normalsize $u_2$};
	\node[anchor=center] at (101pt,39.5pt) {\normalsize $u_m$};
	\node[anchor=center] at (0.5pt,-40.5pt) {\normalsize $w_1$};
	\node[anchor=center] at (30.5pt,-40.5pt) {\normalsize $w_2$};
	\node[anchor=center] at (101pt,-40.5pt) {\normalsize $w_m$};
	
\end{tikzpicture}
	\vspace{5pt}
	\caption{Example system for $\Theta(n^2)$ stabilization time in the monotone model.}
	\label{fig:longstab}
\endminipage\hfill
\hspace{0.07\textwidth}
\minipage{0.6\textwidth}
\centering
	\vspace{12pt}

\begin{tikzpicture}

	\draw[very thick, blue, arrows=-latex] (50pt,20pt) -- (92pt,20pt);
	\draw[very thick, blue, arrows=-latex] (50pt,-20pt) -- (95pt,14pt);
	\draw[very thick, brown, arrows=-latex] (0pt,0pt) -- (42pt,18pt);
	\draw[very thick, brown, arrows=-latex] (0pt,0pt) -- (42pt,-18pt);
	\draw[very thick, blue, arrows=-latex] (50pt,70pt) -- (50pt,28pt);
	\draw[very thick, brown, arrows=-latex] (50pt,70pt) -- (95pt,26pt);
	
	\node[anchor=center] at (20.5pt,15pt) {\small $v_2$};
	\node[anchor=center] at (22.5pt,-17.5pt) {\small $4_{\,}v_1$};
	\node[anchor=center] at (72pt,26pt) {\small $4$};
	\node[anchor=center] at (44pt,45pt) {\small $3$};
	\node[anchor=center] at (83pt,51pt) {\small $6\,u_1$};
	
	\draw[black, fill=white] (0pt,0pt) circle (8pt);
	\draw[black, fill=white] (50pt,20pt) circle (8pt);
	\draw[black, fill=white] (50pt,-20pt) circle (8pt);
	\draw[black, fill=white] (100pt,20pt) circle (8pt);
	\draw[black, fill=white] (50pt,70pt) circle (8pt);

	\node[anchor=center] at (50.5pt,19.5pt) {\large $v_1$};
	\node[anchor=center] at (50.5pt,-20.5pt) {\large $v_2$};
	\node[anchor=center] at (50pt,70pt) {\large $w$};
	
	\draw [fill=white] (2pt,-4pt) rectangle (11pt,-11pt);
	\node[anchor=center] at (6.5pt,-7.5pt) {\scriptsize $\infty$};
	\draw [fill=white] (53.5pt,67pt) rectangle (59.5pt,58pt);
	\node[anchor=center] at (56.5pt,62.5pt) {\footnotesize $3$};

	
	\draw[very thick, brown, arrows=-latex] (180pt,40pt) -- (213pt,23.5pt);
	\draw[very thick, brown, arrows=-latex] (180pt,0pt) -- (213pt,16.5pt);
	\draw[very thick, brown, arrows=-latex] (140pt,20pt) -- (172pt,38pt);
	\draw[very thick, brown, arrows=-latex] (140pt,20pt) -- (172pt,2pt);
	
	\node[anchor=center] at (156pt,36.5pt) {\small $u_2$};
	\node[anchor=center] at (157pt,3.5pt) {\small $u_1$};
	\node[anchor=center] at (201pt,36pt) {\small $v_2$};
	\node[anchor=center] at (201pt,3pt) {\small $v_1$};
	
	\draw[black, fill=white] (140pt,20pt) circle (8pt);
	\draw[black, fill=white] (180pt,40pt) circle (8pt);
	\draw[black, fill=white] (180pt,0pt) circle (8pt);
	\draw[black, fill=white] (220pt,20pt) circle (8pt);

	\node[anchor=center] at (180.5pt,39.5pt) {\large $u_1$};
	\node[anchor=center] at (180.5pt,-0.5pt) {\large $u_2$};
	
	\draw [fill=white] (142pt,16pt) rectangle (151pt,9pt);
	\node[anchor=center] at (146.5pt,12.5pt) {\scriptsize $\infty$};
	
\end{tikzpicture}
	\caption{Example system where early defaulting is the best strategy in the monotone model.}
	\label{fig:earlydef_mon}
\endminipage\hfill
\end{figure}

\subsection{Defaulting strategies}

Finally, we discuss how the monotone model compares to the reversible model in terms of defaulting strategies.

When finding the globally best ordering, the two models turn out to be very similar. In fact, our proofs from Section \ref{sec:reversible_best} can also be carried over to the monotone model without any changes.

\begin{corollary}
Lemma \ref{lem:branching_best} and Theorems \ref{th:best_ordering_hard} and \ref{th:worst_ordering_hard} also hold in the monotone model.
\end{corollary}

In terms of individual defaulting strategies, the branching gadget again provides a simple example where late defaulting is beneficial: by delaying their updates, banks $u$ and $v$ can again entirely avoid a default.

However, early defaulting is a more difficult question in this setting. In particular, we cannot hope for a result that is analogous to Theorem \ref{th:earlydef_rev}, since once a bank reports a default, there is no way to reverse this in the future. Nonetheless, early defaulting can still be a beneficial strategy in the monotone model: there are cases when a bank cannot avoid an eventual default in any way, but early defaulting can still allow the bank to have a higher recovery rate in the final state.

\begin{theorem} \label{th:earlydef_mon}
There exists a system where a bank $v$ only obtains its highest possible recovery rate in the final state of the system if $v$ is the first bank to report a default.
\end{theorem}

\begin{proof}

Consider the system in Figure \ref{fig:earlydef_mon}, where either $v_1$ or $v_2$ can execute the first update. We analyze the defaulting strategies of bank $v_1$ in this system.

Assume that $v_1$ is the first to execute a step, announcing $r_{v_1}=\frac{3}{4}$. This gives new assets to $v_2$ (resulting in $\overline{r_{v_2}}=1$), and new liabilities to $u_2$ (resulting in $\overline{r_{u_2}}=0$). The next update can only be executed by $u_2$, setting $r_{u_2}=0$; at this point, the system stabilizes.

On the other hand, assume that $v_2$ first announces $r_{v_2}=0$. This provides $\overline{r_{v_1}}=1$ and $\overline{r_{u_1}}=0$, so as a next step, $u_1$ will announce a default. However, this results in new liabilities for $w$, so as a next step, $w$ has to update to $r_w=\frac{1}{3}$. With this, $v_1$ only has $2$ assets altogether, so $v_1$ must announce $r_{v_1}=\frac{1}{2}$. Hence $v_1$ achieves a lower recovery rate in the final state if it is not the first bank to announce a default.

Note that with some further modifications, we can also make the example symmetric to ensure that both $v_1$ and $v_2$ are motivated to be the first one to default. \qedhere
\end{proof}

Finally, one might also wonder if the monotone model allows an analogous result to Theorem \ref{th:findbesttime}, i.e. a hardness result on finding the best defaulting strategy of a single bank. However, note that the simple formulation of Theorem \ref{th:findbesttime} was possible due to the fact that the proof construction only allowed one possible ordering in the rest of the system.

If we were to introduce a similar setting in the monotone model, then the banks could always find the best outcome in polynomial time, since the sequence can only last for $O(n^2)$ steps. As such, in the monotone model, we can only expect similar hardness results for more complex formulations of this problem, such as finding the best defaulting time with respect to, e.g., the best-case or worst-case ordering of the remaining banks in the system.


\bibliography{references}

\newpage

\begin{appendices}

\section{Binary counter construction} \label{App:A}

In this section, we describe the binary counter construction that proves Theorem \ref{th:binary_counter}.

\subsection{Stable bit gadget}

One of the basic building blocks of this construction is the so-called \textit{stable bit gadget}, shown in Figure \ref{fig:stable_base}; we have already applied the base idea of this gadget in the proof of Theorem \ref{th:earlydef_rev}. The gadget consists of two nodes $v_1$ and $v_2$, which have an outgoing CDS in reference to each other. Note that if some external condition sets $r_{v_1}$ to $0$, then this results in $r_{v_2}=0$, even if we had $r_{v_2}=1$ before. Similarly, if we change to $r_{v_1}=1$, then this results in $r_{v_2}=1$, even if we had $r_{v_2}=0$ before.

The key property of this gadget is that it allows us to ensure that banks $v_1$ and $v_2$ remain in a specific state. Assume that we initially have $v_1$ and $v_2$ in the state $r_{v_1}=r_{v_2}=1$, and some event (i.e. the default of an external node) creates another liability for $v_1$. This leads to $r_{v_1}=0$, and hence $r_{v_2}=0$. However, after this point, even if the extra liabilities for $v_1$ are removed, the banks $v_1$ and $v_2$ do not return to their initial recovery rate, but remain in this new state of $r_{v_1}=r_{v_2}=0$ instead.

\begin{figure}[b]
\centering
\minipage{0.32\textwidth}
\centering

\begin{tikzpicture}
	
	\draw[very thick, brown, arrows=-latex] (0pt,25pt) -- (43pt,3.5pt);
	\draw[very thick, brown, arrows=-latex] (0pt,-25pt) -- (43pt,-3.5pt);
	
	\node[anchor=center] at (27pt,19pt) {\small $v_2$};
	\node[anchor=center] at (27pt,-20pt) {\small $v_1$};
	
	\draw[black, fill=white] (0pt,25pt) circle (8.5pt);
	\draw[black, fill=white] (0pt,-25pt) circle (8.5pt);
	\draw[black, fill=white] (50pt,0pt) circle (8.5pt);
	
	\node[anchor=center] at (0.5pt,24.5pt) {\large $v_1$};
	\node[anchor=center] at (0.5pt,-25.5pt) {\large $v_2$};

\end{tikzpicture}
	\vspace{4pt}
	\caption{Stable bit gadget.}
	\label{fig:stable_base}
\endminipage\hfill
\hspace{0.03\textwidth}
\minipage{0.57\textwidth}
\centering
	\vspace{9pt}

\begin{tikzpicture}
	
	\draw[very thick, brown, arrows=-latex] (0pt,0pt) -- (33pt,17pt);
	\draw[very thick, brown, arrows=-latex] (0pt,-0pt) -- (33pt,-17pt);
	\draw[very thick, brown, arrows=-latex] (40pt,20pt) -- (55pt,26pt) -- (80pt,26pt) -- (92pt,21pt);
	\draw[very thick, brown, arrows=-latex] (40pt,20pt) -- (55pt,14pt) -- (80pt,14pt) -- (92pt,19pt);
	\draw[very thick, brown, arrows=-latex] (40pt,-20pt) -- (55pt,-26pt) -- (80pt,-26pt) -- (92pt,-21pt);
	\draw[very thick, brown, arrows=-latex] (40pt,-20pt) -- (55pt,-14pt) -- (80pt,-14pt) -- (92pt,-19pt);
	
	\node[anchor=center] at (17pt,15pt) {\small $z_1$};
	\node[anchor=center] at (17pt,-15pt) {\small $z_1$};
	\node[anchor=center] at (67.5pt,30.5pt) {\small $v_2$};
	\node[anchor=center] at (67.5pt,9pt) {\small $z_2$};
	\node[anchor=center] at (67.5pt,-8.5pt) {\small $v_1$};
	\node[anchor=center] at (67.5pt,-31pt) {\small $z_2$};
	
	\draw[black, fill=white] (100pt,20pt) circle (8pt);
	\draw[black, fill=white] (100pt,-20pt) circle (8pt);
	\draw[black, fill=white] (40pt,20pt) circle (8pt);
	\draw[black, fill=white] (40pt,-20pt) circle (8pt);
	\draw[black, fill=white] (0pt,0pt) circle (8pt);
	
	\node[anchor=center] at (40.5pt,19.5pt) {\large $v_1$};
	\node[anchor=center] at (40.5pt,-20.5pt) {\large $v_2$};
	
	\draw [fill=white] (2pt,-5pt) rectangle (11pt,-12pt);
	\node[anchor=center] at (6.5pt,-8.5pt) {\scriptsize $\infty$};
	
\end{tikzpicture}
	\vspace{2pt}
	\caption{Resettable version of the stable bit gadget (the sink node is split into two for a cleaner topology).}
	\label{fig:stable_reset}
\endminipage\hfill
\end{figure}

Hence if we add conditional assets and liabilities to both nodes of the gadget, then activating these contracts will allow us to flip the state of the bit to $0$ or $1$ as required, and then the gadget will store this state until the next such activation. More specifically, consider the extended version of the gadget in Figure \ref{fig:stable_reset}. The default state of the external nodes $z_1$, $z_2$ is $r_{z_1}=1$ and $r_{z_2}=1$; in this case, $v_1$ and $v_2$ are in the same situation as in Figure \ref{fig:stable_base}, so they retain their current recovery rates. However, if we set $r_{z_1}=0$ and $r_{z_2}=1$, then this allows us to set the gadget to $r_{v_1}=r_{v_2}=1$, regardless of its previous state. Similarly, if we have $r_{z_1}=1$ and $r_{z_2}=0$, then this allows us to set $r_{v_1}=r_{v_2}=0$, regardless of the previous state. Our construction ensures that after each such operation, the external nodes $z_1, z_2$ are returned to their default state $r_{z_1}=r_{z_2}=1$.

Thus the gadget essentially acts as a memory cell for storing a single bit, which is modifiable through the recovery rates of external banks. Our main construction will use such gadgets to store the current bits of the binary counter, and apply these external operations to increment the counter to the next value.

\subsection{States and conditions}

In order to ensure that the bits are changed in the correct order for the incrementation, we create another set of gadgets that capture the current state of the counting process, and allows us to control the transitions between these states.

First, note that CDSs essentially allow us to describe specific conditions, and ensure that an event only happens if these conditions are fulfilled. Assume that we have some nodes $z_1, ..., z_c$ and $z'_1, ..., z'_d$, which are all `binary nodes' in the sense that the system guarantees that they always have a recovery rate of either $0$ or $1$. Let us first analyze the left-hand component of the system in Figure \ref{fig:state} separately; we will refer to this building block as the \textit{condition gadget}. In this gadget, node $u$ has incoming CDSs in reference to banks $z_1, ..., z_c$, and outgoing CDSs in reference to banks $z'_1, ..., z'_d$, and an outgoing debt of weight $c$. Furthermore, we have another node $w$ with a liability of $1$, and an incoming CDS in reference to $u$ which has a weight of $c+1$.

The key property of this gadget is that it only allows $r_w=0$ if $r_{z_1}=...=r_{z_c}=0$ and $r_{z'_1}=...=r_{z'_d}=1$. That is, if all these conditions are fulfilled, then $u$ has $c$ assets and $c$ liabilities, thus $r_u=1$ and $r_w=0$. However, if any of the nodes $z_i$ have $r_{z_i}=1$, then $u$ has at most $c-1$ assets and $r_u \leq \frac{c-1}{c}$. This ensures that $w$ receives a payment of at least $\frac{1}{c} \cdot (c+1) \geq 1$, thus avoiding default with $r_w=1$. Similarly, if any of the nodes $z'_i$ have $r_{z'_i}=0$, then $u$ has at least $c+1$ liabilities and $r_u \leq \frac{c}{c+1}$. This again means that $w$ gets a payment of at least $\frac{1}{c+1} \cdot (c+1) = 1$, thus ensuring $r_w=1$ again.

Hence this gadget allows us to select a set of binary banks, and ensure that $w$ only goes into default if each of these banks have the desired recovery rate. This allows us to control the transitions between a set of states, only permitting entry into a state if a set of conditions are fulfilled.

\begin{figure}
\centering

\begin{tikzpicture}
	
	\draw[very thick, brown, arrows=-latex] (50pt,0pt) -- (70pt,4pt) -- (123pt,4pt) -- (143pt,1pt);
	\draw[very thick, brown, arrows=-latex] (50pt,0pt) -- (70pt,12pt) -- (123pt,12pt) -- (144pt,3.5pt);
	\draw[very thick, brown] (50pt,0pt) -- (70pt,-4pt) -- (90pt,-4pt);
	\draw[very thick, brown, arrows=-latex] (110pt,-4pt) -- (123pt,-4pt) -- (143pt,-1pt);
	\draw[very thick, brown, arrows=-latex] (50pt,0pt) -- (70pt,-12pt) -- (123pt,-12pt) -- (144pt,-3.5pt);
	
	\draw[very thick, brown, arrows=-latex] (150pt,0pt) -- (170pt,4pt) -- (223pt,4pt) -- (243pt,1pt);
	\draw[very thick, brown, arrows=-latex] (150pt,0pt) -- (170pt,12pt) -- (223pt,12pt) -- (244pt,3.5pt);
	\draw[very thick, brown] (150pt,0pt) -- (170pt,-4pt) -- (190pt,-4pt);
	\draw[very thick, brown, arrows=-latex] (210pt,-4pt) -- (223pt,-4pt) -- (243pt,-1pt);
	\draw[very thick, brown, arrows=-latex] (150pt,0pt) -- (170pt,-12pt) -- (223pt,-12pt) -- (244pt,-3.5pt);
	
	\draw[very thick, blue, arrows=-latex] (150pt,0pt) -- (170pt,-28pt) -- (223pt,-28pt) -- (248pt,-7pt);
	
	\draw[very thick, brown, arrows=-latex] (50pt,0pt) -- (50pt,-50pt) -- (142pt,-50pt);
	\draw[very thick, blue, arrows=-latex] (150pt,-50pt) -- (250pt,-50pt) -- (250pt,-8pt);
	
	\node[anchor=center] at (90pt,16pt) {\footnotesize $z_1$};
	\node[anchor=center] at (100pt,7.5pt) {\footnotesize $z_2$};
	\node[anchor=center] at (100pt,-4pt) {\large ...};
	\node[anchor=center] at (120pt,-16pt) {\footnotesize $z_c$};
	
	\node[anchor=center] at (190pt,17pt) {\footnotesize $z'_1$};
	\node[anchor=center] at (200pt,8.5pt) {\footnotesize $z'_2$};
	\node[anchor=center] at (200pt,-3pt) {\large ...};
	\node[anchor=center] at (220pt,-15pt) {\footnotesize $z'_d$};
	
	\node[anchor=center] at (100pt,-43pt) {\small $(c+1) \, u$};
	
	\node[anchor=center] at (197pt,-33pt) {\small $c$};
	
	\draw[black, fill=white] (50pt,0pt) circle (8.5pt);
	\draw[black, fill=white] (150pt,0pt) circle (8.5pt);
	\draw[black, fill=white] (250pt,0pt) circle (8.5pt);
	
	\draw[black, fill=white] (150pt,-50pt) circle (8.5pt);
	
	\node[anchor=center] at (150pt,0pt) {\large $u$};
	\node[anchor=center] at (150pt,-50pt) {\large $w$};
	
	\draw [fill=white] (52pt,-5pt) rectangle (61pt,-12pt);
	\node[anchor=center] at (56.5pt,-8.5pt) {\scriptsize $\infty$};
	
	
	\draw[very thick, brown, arrows=-latex] (320pt,0pt) -- (335pt,6pt) -- (360pt,6pt) -- (372pt,1pt);
	\draw[very thick, brown, arrows=-latex] (320pt,0pt) -- (335pt,-6pt) -- (360pt,-6pt) -- (372pt,-1pt);
	\draw[very thick, brown, arrows=-latex] (320pt,-40pt) -- (380pt,-40pt) -- (380pt,-8pt);
	
	\node[anchor=center] at (347.5pt,11pt) {\small $w$};
	\node[anchor=center] at (347.5pt,-12pt) {\small $v_2$};
	\node[anchor=center] at (347.5pt,-35pt) {\small $v_1$};
	
	\draw[black, fill=white] (380pt,0pt) circle (8pt);
	\draw[black, fill=white] (320pt,0pt) circle (8pt);
	\draw[black, fill=white] (320pt,-40pt) circle (8pt);
	
	\node[anchor=center] at (320.5pt,-0.5pt) {\large $v_1$};
	\node[anchor=center] at (320.5pt,-40.5pt) {\large $v_2$};

\end{tikzpicture}
	\caption{State gadget, obtained as the combination of a condition gadget (left) and a variant of the stable bit gadget (right).}
	\label{fig:state}
\end{figure}
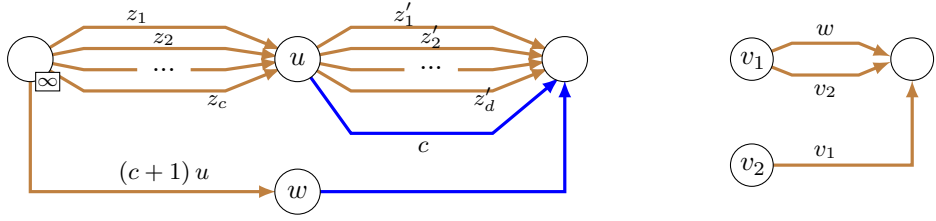

We can combine this condition gadget with a stable bit gadget to obtain our \textit{state gadget} as shown in Figure \ref{fig:state}. The state gadget ensures that a specific set of conditions are fulfilled before allowing $w$ to default. This then sets the stable bit to $0$, representing the fact that the execution is currently in this state; we can then use bank $v_2$ as a reference entity in the CDSs of other condition gadgets to make some events dependent on the condition that we are currently in this state. Once we exit the state, we can use the technique shown in Figure \ref{fig:stable_reset} to ensure that the state bit is set back to $1$ again. Also, note that since the bit is stable, once we enter the state, the bit remains set to $0$ until it is artificially reset with this technique, even if the entry condition of the state becomes false in the meantime (and thus $w$ updates back to $r_w=1$).

Hence the state gadgets will allow us to represent the execution of the counting process through a set of states and transitions between these states, with the next state always depending on our current state and possibly the value of some other stable bits in the system (the bits of the counter, in our current case). By defining the appropriate conditions (including a specific previous state) for each state, we can ensure that the only valid ordering of the system is an execution that follows these prescribed transitions between the states.

\subsection{Resetting the states}

In order to guarantee that the system is only in one state at any point in time, we also have to ensure that the resetting operations are indeed executed after exiting a state (i.e. that $r_{v_1}$ and $r_{v_2}$ is indeed updated to $1$). Due to this, encoding the transition from a state $s_1$ to another state $s_2$ becomes a nontrivial task. The natural approach would be to enforce the resetting of $s_1$ by including these updates in the entry condition of $s_2$. However, recall that the entry condition of $s_2$ also requires that $s_1$ is active (as a preceding state), so this makes the entry condition of $s_2$ contradictory, hence impossible to fulfill. On the other hand, after entering state $s_2$, there is no straightforward way to verify the resetting of $s_1$ anymore; furthermore, the system already has two active states at once in this case.

In order to solve this problem, we take each state $s$ of our original system design, and replace it by three consecutive state gadgets $\textsc{entry}_s$, $\textsc{reset}_s$ and $\textsc{exit}_s$, known as the \textit{entry phase}, \textit{resetting phase} and \textit{exit phase} of $s$. State $\textsc{entry}_s$ will have the same entry conditions as the original state $s$ did, and any other state that was previously following state $s$ will now follow after state $\textsc{exit}_s$. The entry condition for states $\textsc{reset}_s$ and $\textsc{exit}_s$ will be that we are currently in states $\textsc{entry}_s$ and $\textsc{reset}_s$, respectively. Thus instead of passing through state $s$, the execution will pass through all $3$ phases of $s$ in this predefined order in our modified system.

The key idea is that the three classes of states will reset each other in a round-robin fashion throughout the execution. State $\textsc{reset}_s$ will provide new assets to the nodes $v_1$ and $v_2$ of state $\textsc{entry}_s$, thus resetting their recovery rate to $1$. Then state $\textsc{exit}_s$ will have it as an entry condition that the value of the banks $w$, $v_1$ and $v_2$ of state $\textsc{entry}_s$ are all set to $1$. This ensures that $\textsc{exit}_s$ is indeed only reached when all recovery rates in $\textsc{entry}_s$ are set back to their initial value. This way we can make sure that when the execution leaves $\textsc{exit}_s$ and enters the state $\textsc{entry}_{s'}$ of the following state $s'$, then the state $s$ is not considered active anymore.

In order to ensure that $\textsc{reset}_s$ and $\textsc{exit}_s$ are also reset to inactive, we execute the same steps in any two succeeding states for both of them. That is, we ensure that $\textsc{exit}_s$ will provide new assets to reset the stable bit of $\textsc{reset}_s$ (in the same fashion that $\textsc{reset}_s$ does this for $\textsc{entry}_s$), and we ensure that in any original state $s'$ succeeding $s$, the state $\textsc{entry}_{s'}$ has in its entry condition that banks $w$, $v_1$ and $v_2$ of state $\textsc{reset}_s$ are all set to $1$ (in the same fashion that $\textsc{exit}_s$ does this for $\textsc{entry}_s$). Similarly, in order to ensure that $\textsc{exit}_s$ is reset to inactive, we take every original state $s'$ succeeding $s$, and in $\textsc{entry}_{s'}$ we provide new assets to reset the stable bit of $\textsc{exit}_s$, while we require in the entry condition of $\textsc{reset}_{s'}$ that banks $w$, $v_1$ and $v_2$ of $\textsc{exit}_s$ are set to $1$.

Hence by representing each logical state by three consecutive state gadgets, we can ensure that any ordering of updates is indeed forced to reset each state to inactive when leaving the state and entering the following one. Note that this method results in a higher number of entry conditions for each of our state gadgets, but this has no effect on the overall construction. Furthermore, we point out that the reason to have at least three such classes is to avoid the situation when a state gadget investigates its own banks in its entry condition, which could lead to undesired behavior.

\subsection{Technical details of managing states}

While this already describes the general technique of using state gadgets, there are still several details to discuss for completeness.

One such example is the handling of bank $w$ in the state gadget: note that this bank is slightly differently from $v_1$ and $v_2$ in the sense that it does not receive extra assets to be reset, but its resetting is still checked as a condition in $\textsc{exit}_s$. This is because when state $\textsc{reset}_s$ is reached, then the exit state $\textsc{exit}_{\hat{s}}$ of the previous state $\hat{s}$ has already been reset to inactive, which means that the entry conditions of $\textsc{entry}_s$ are not fulfilled anymore. This allows us to update $r_u$ to a new value of $r_u \leq \frac{c}{c+1}$, which then allows us to set $r_w=1$, and indeed enter $\textsc{exit}_s$. Hence by not explicitly resetting $w$ but still checking in $\textsc{exit}_s$ that $w$ is reset, we can ensure that bank $u$ of state $s$ is also reset to its initial state of $r_u \leq 1$, and thus the $s$ state can only be reactivated if the entry conditions are fulfilled again at some point.

Furthermore, note that while our construction mostly requires us to implement logical `and' relations in the entry conditions of state gadgets, we occasionally also have to implement a logical `or'. One such case is the central state of our binary counter which will have multiple different preceding states, i.e. there are multiple states that finish by enabling this state as the next one; to activate this state, we only require one of the preceding states to be active, and not all of them. In this specific case, it is rather simple to insert this `or' condition into our state gadget: we simply add an incoming CDS in reference to each of these preceding states, but we still select the weight of the input CDS of $w$ based on the original $c$ value, i.e. as if there was only one preceding state. This way $u$ receives a payment of $1$ if we are in any of these preceding states (and the remaining conditions are fulfilled), and since we can not have two active states at the same time, these extra CDSs will always only result in a payment of $0$ or $1$ for $u$ altogether. In a more general setting (e.g. if we want to encode different further conditions for different predecessor states), we can create a separate transition state for each such condition, and then use the same method to set these transition states as predecessors.

Finally, our analysis has so far assumed that inactive states gadgets have recovery rates of $r_u<1$ and $r_w=1$, and we discussed how we can maintain this invariant throughout the process. However, we also have to ensure this when initializing the construction; since we initially begin with $r_u=1$, the node $w$ of any state gadget could already report a default in the initial time step, even though the entry conditions of the state are not satisfied.

For this initialization step (i.e. achieving $r_u<1$ in each state gadget), we introduce a special node $y_0$ whose default indicates that the system has already been initialized, i.e. that $r_u<1$ in each state gadget. Then we slightly modify the state gadgets such that the outgoing debt from node $w$ is replaced by a CDS in reference to $y_0$; this way no $w$ can update before all the banks $u$ are initialized, but after we set $r_{y_0}=0$, the recovery rate of $y_0$ will never change, and thus the outgoing contract of each node $w$ will behave as a single debt for the rest of the process.

To ensure this behavior, it suffices to add a starter node $y_1$ with an outgoing CDS in reference to each $u$, and carefully choose $e_{y_1}$ such that $y_1$ only goes into default if each bank $u$ has executed an update. Then we can include $y_0$ in a stable bit gadget with another node $y'_0$, with $y'_0$ also having an outgoing CDS in reference to $y_1$. This system can only begin with all the state gadget nodes $u$ executing an update. The slight default of $u$ then also sends $y'_0$, and then $y_0$ into default; since $y_0$ has no way to reverse this, it will remain in default indefinitely. We can then use the default of $y_0$ as the trigger condition for the first real state of our counting process.

\subsection{Overall construction}

Given the tools to represent states, bits and conditions, the construction of the binary counter becomes straightforward. Let us introduce a parameter $k$ such that $k=\Theta(n)$. We create $k$ distinct stable bit gadgets that represent the $k$ bits of a counter, which will count from $0$ to $2^k-1$. We also add a set of state gadgets which control the counting process. More specifically, we add an \textsc{idle} state to capture the state between two consecutive incrementations. Furthermore, for each bit of the counter (i.e. each $i \in \{ 1, ..., k \}$), we add two states that describe the incrementation of the $i^{\text{th}}$ bit, and we call them \textsc{enable}$_i$ and \textsc{done}$_i$. The states of our process are illustrated in Figure \ref{fig:counter_states}.

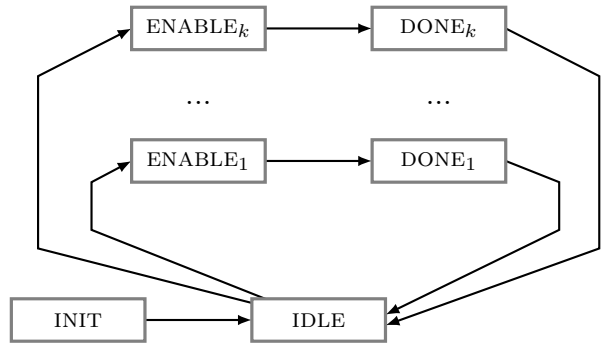
\begin{figure}
\centering

\begin{tikzpicture}
	
	\draw[thick, black, arrows=-latex] (50pt,8pt) -- (90pt,8pt);
	
	\draw[thick, black, arrows=-latex] (95pt,68pt) -- (135pt,68pt);
	\draw[thick, black, arrows=-latex] (95pt,118pt) -- (135pt,118pt);
	
	\draw[thick, black, arrows=-latex] (115pt,8pt) -- (30pt,42pt) -- (30pt,60pt) -- (45pt,68pt);
	\draw[thick, black, arrows=-latex] (115pt,8pt) -- (10pt,35pt) -- (10pt,100pt) -- (45pt,118pt);
	
	\draw[thick, black, arrows=-latex] (185pt,68pt) -- (205pt,60pt) -- (205pt,42pt) -- (140pt,10pt);
	\draw[thick, black, arrows=-latex] (185pt,118pt) -- (220pt,100pt) -- (220pt,35pt) -- (140pt,6pt);
	
	\node[anchor=center] at (70pt,90pt) {\Large ...};
	\node[anchor=center] at (160pt,90pt) {\Large ...};
	
	\draw [gray, very thick, fill=white] (0pt,0pt) rectangle (50pt,16pt);
	\node[anchor=center] at (25pt,8pt) {\normalsize \textsc{init}};
	
	\draw [gray, very thick, fill=white] (90pt,0pt) rectangle (140pt,16pt);
	\node[anchor=center] at (115pt,8pt) {\normalsize \textsc{idle}};
	
	\draw [gray, very thick, fill=white] (45pt,60pt) rectangle (95pt,76pt);
	\node[anchor=center] at (70pt,68pt) {\normalsize \textsc{enable}$_1$};
	
	\draw [gray, very thick, fill=white] (45pt,110pt) rectangle (95pt,126pt);
	\node[anchor=center] at (70pt,118pt) {\normalsize \textsc{enable}$_k$};
	
	\draw [gray, very thick, fill=white] (135pt,60pt) rectangle (185pt,76pt);
	\node[anchor=center] at (160pt,68pt) {\normalsize \textsc{done}$_1$};
	
	\draw [gray, very thick, fill=white] (135pt,110pt) rectangle (185pt,126pt);
	\node[anchor=center] at (160pt,118pt) {\normalsize \textsc{done}$_k$};

\end{tikzpicture}
	\caption{Illustration of the main states of our binary counter system.}
	\label{fig:counter_states}
\end{figure}

For \textsc{enable}$_i$, the entering condition is that the process is currently in the \textsc{idle} state, and that this is indeed a valid next incrementation of the counter, i.e. that the bits at positions $1, ..., i-1$ are all set to $1$, and the bit at position $i$ is set to $0$. When entering \textsc{enable}$_i$, we ensure that this state sets the bits at positions $1, ..., i-1$ to $0$, and it sets the bit at position $i$ to $1$. The entering condition of the \textsc{done}$_i$ state is that all of these updates are indeed executed, and thus the counter is indeed correctly incremented. From the \textsc{done}$_i$ state, we lead the execution back to the \textsc{idle} state without any condition.

Thus the financial system indeed has essentially only one possible ordering, aside from the fact that we are free to choose the order of updating the bits of the counter in each incrementation. In this single ordering, the system works as a binary counter: in the idle state, the only valid next step is to always execute the next incrementation on the counter. Since the construction consists of only $O(k)$ different gadgets, each having only $O(1)$ nodes, this indeed allows for a choice of $k=\Theta(n)$, and thus the counting process indeed lasts for at least $2^{\Omega(n)}$ steps (note that this also holds if we only count defaulting steps, i.e. when a bank $v$ updates from $r_v=1$ to $r_v<1$). Once all the bits are set to $1$, there is no next state that the process can enter from the \textsc{idle} state, so the system indeed stabilizes eventually.

As a technical detail, note that we must also ensure that the process begins in the \textsc{idle} state. This can be ensured by adding a further state \textsc{init} and a further stable bit gadget with this state. The state \textsc{init} can be entered if this stable bit is set to $1$ (and the initialization node $y_0$ has already defaulted), so it will be the only state that the process can enter in the beginning. We allow the process to also enter the \textsc{idle} state from \textsc{init}. However, within the \textsc{init} state, this stable bit is set to $0$, and the system does not provide a way for this stable bit to ever be reset to $1$; hence the \textsc{init} state cannot ever be entered again, and thus it plays no role after this point.

Furthermore, note that for the simplest implementation, we can consider the counter bits to be $1$ when the nodes of the stable bit are in default, and $0$ when they are not in default: this ensures that the initial state of the system (when every bank has a recovery rate of $1$) is indeed a valid initialization of our construction. Otherwise (if we want the $0$ bits to be represented by defaulting bit gadgets), we can use further initial states to ensure that each stable bit is initialized to the desired value before the process first enters the \textsc{idle} state.

Recall that, as discussed before, each state in our description will in fact be split to three consecutive states to ensure that resetting is always executed. However, this only increases the number of state gadgets in our system by a constant factor, and thus it has no effect on our analysis.

\section{Best defaulting strategy for a single bank} \label{App:B}

In this section we prove the claim of Theorem \ref{th:findbesttime}, i.e. that finding the best defaulting strategy (the best time to report a default) for a single bank is an NP-hard problem. We combine the binary counter construction of Appendix \ref{App:A} with the MAXSAT reduction technique to show that any efficient algorithm that finds the best time to report a default would also provide an efficient solution to MAXSAT.

\subsection{Overall idea}

The main idea of our construction is to create a binary counter system where each counter bit represents one of the variables of our input MAXSAT formula. The counting process then corresponds to enumerating all the $2^k$ possible value assignments to the variables.

We then add a further node $v$ to this system that wants to find the best time to report its own default. This bank $v$ will only have an opportunity to report a default in the \textsc{idle} state of the counter, i.e. exactly once for each of the $2^k$ possible assignments. When $v$ reports a default, this will immediately stop the counting process, thus fixing the value of the $k$ stable bits to their current value forever. Then for each clause of the formula, we add a clause variable that defaults exactly if at least one of the bits corresponding to the literals in the clause are set to true (but only after the counting has stopped). Finally, we ensure that for each such clause node, bank $v$ receives a unit of payment through an incoming CDS.

This results in a construction where, by choosing a time to report its default, $v$ can essentially select a value assignment to the variables, and the final amount of assets received by $v$ will be determined by the number of satisfied clauses under this assignment. Thus selecting the best defaulting time for $v$ is equivalent to selecting the best assignment for MAXSAT, which completes the reductions.

Note that the behavior of the binary counter system is completely predictable, since it only has essentially one valid ordering of updates, but it is still NP-hard to find the optimal defaulting time for $v$. This implies similar hardness results in more general systems that have very different orderings: it is still NP-hard to find the best defaulting time if we, for example, assume that the remaining part of the systems follows the ordering that is the most/least beneficial for $v$.

Furthermore, we note that in order to simplify our clause gadgets, we can easily extend the binary counter construction by a negated version of each of the $k$ stable bits, which are similarly set and checked in the \textsc{enable}$_i$ and \textsc{done}$_i$ states. This step essentially provides another counter that is counting backwards from $2^k-1$ to $0$ simultaneously to our original counter, without having any effect on the magnitude of the number of nodes. More importantly, in our case, it provides a convenient access to the negation of each of variable; with this, we can directly check the value of any literal in the clauses of the formula.

\subsection{Technical details}

Consider the bank $v$ for which we want to find the best defaulting strategy. In order to ensure that $v$ can only report a default when the counter is in the \textsc{idle} state, we simply add an outgoing CDS of weight $1$ to $v$ (and select $e_v=0$). Then similarly to the design of a state gadget, we connect $v$ to a stable bit gadget on banks $w_1$ and $w_2$ such that the default of $v$ will lead to the default of both $w_1$ and $w_2$ (i.e. $e_{w_1}=e_{w_2}=0$, and $w_1$ has outgoing CDSs in reference to $v$ and $w_2$, while $w_2$ has an outgoing CDS in reference to $w_1$). Then $w_1$ and $w_2$ can never return from this default; this will ensure that even after $v$ receives extra assets in the future, the counting still does not continue.

More specifically, the outgoing CDS of $v$ is in reference to bank $v_2$ of the $\textsc{entry}_{\textsc{idle}}$ state; since this bank defaults every time when the \textsc{idle} state is visited, $v$ indeed has the opportunity to report a default at each of the $2^k$ counting phases. Then in the $\textsc{exit}_{\textsc{idle}}$ state, we add it as a condition that both $r_v=1$ and $r_{w_2}=1$; this ensures that after $v$ defaults, the counter can never enter the $\textsc{exit}_{\textsc{idle}}$ state again, so the counting indeed stops at the current value.

Furthermore, for each literal of the formula (i.e. each variable and its negation), we create a node that defaults in this final state if the literal is set to true. That is, given a literal $\ell_i$ (a variable or its negated version), we add a bank $\ell_i$ representing this literal. This bank $\ell_i$ has an outgoing CDS in reference to $w_2$, and an incoming CDS in reference to the stable bit gadget in the counter that represents the negated version of $\ell_i$. This implies that (i) the bank $\ell_i$ can only go into default once $v$ has reported a default and the counter was stopped, and (ii) in this case, it goes into default exactly if the literal $\ell_i$ is set to true in the chosen assignment.

Then for each clause $c_i$ of the formula, we simply add a bank representing $c_i$, and draw an outgoing CDS from $c_i$ for all the banks $\ell_i$ that correspond to a literal included in the clause. As such, $c_i$ is in default exactly if at least one of the literals in the clause is true.

Finally, for each such clause node $c_i$, we add an incoming CDS to $v$ in reference to this bank $c_i$. With this, the incoming assets of $v$ equal the number of satisfied clauses under the chosen assignment. Furthermore, we add another incoming CDS in reference to $w_2$ in order to compensate for the outgoing CDS that allows $v$ to default. With this, the total payoff of $v$ in the final state (the difference of its assets and liabilities) is indeed equal to the number of satisfied clauses in the formula. Hence the best outcome for $v$ is indeed the assignment where the maximal possible number of clauses are satisfied, which is NP-hard to find.

Hence whenever $v$ reports a default, the connected stable bit is set to $0$, and the counting stops permanently. After this point, the literal gadgets corresponding to true literals will have to report a default eventually, followed by the appropriate clause gadgets. Hence the system will indeed eventually stabilize in the state where $v$ receives all the payments for the clause gadgets, so its assets in the final state are indeed defined by the quality of the MAXSAT assignment. This completes our reduction.

\section{Financial networks as a model of computation} \label{App:TM}

We now make a detour to briefly discuss how financial systems can behave as a model of computation. While this is not very relevant for the analysis of real-world financial networks, it is still an interesting aspect of our reversible model from a theoretical perspective.

First of all, note that we have only considered financial networks with finitely many banks; as such, in our base model of these systems, we cannot hope to model a general Turing machine (TM) with an infinitely long tape. Therefore, we only discuss how our networks can model a Turing machine which only has a finitely long tape (also known as a linear bounded automaton). We point out that generalizing our constructions to the infinite case would not be as straightforward as to simply allow infinitely many banks and contracts in the network. For example, in our binary counter or in the TM simulation design below, this generalization would lead to infinitely many state gadgets, and the initialization of these gadgets would already require infinitely many updates in the beginning, thus not allowing the main part of the process to begin after a finite amount of steps.

Recall that the main tools for simulating a TM have already been introduced in the binary counter construction: state machines were explicitly discussed and used in the counter, and the stable bit gadgets are a natural candidate to simulate the tape cells of a TM over a binary alphabet. Note that for a very simple and crude encoding of the TM, the stable bit gadgets are not even needed: since a TM with a finite tape can only have finitely many valid configurations (in terms of current state, tape content and tape pointer position), we can encode the transitions between these configurations as a finite automaton, and build this state machine using our state gadgets.

However, a much more elegant way of modeling a Turing machine with our systems is to indeed use stable bits as tape cells, and encode an addressing mechanism on the tape. That is, besides our financial subsystem representing the state machine part of the TM, we also create a binary counter which stores the position of the TM pointer of the tape; when the pointer is moved to the left or the right in a transition from one state to another, we simply increment or decrement the value of this counter.

Then in an auxiliary state following the transition, we can use the value of the counter and the content of the tape to copy the content of the currently chosen tape cell to a specific stable bit gadget. That is, for each cell $c$ of the tape, we have two specific states \textsc{read}$_{c,0}$ and \textsc{read}$_{c,1}$ that are only entered as a next step if the counter value currently points to $c$; state \textsc{read}$_{c,0}$ is activated if stable bit gadget of $c$ is currently set to $0$, while \textsc{read}$_{c,1}$ is activated if $c$ is set to $1$. We use this extra state to copy the content of the cell to a specific stable bit gadget, and then we only make our next transition in the state machine based on the current logical state and the value of this single stable bit (as in case of a Turing machine).

We can use a similar technique for overwriting the value in the current cell: we add two auxiliary writing states \textsc{write}$_{c,0}$ and \textsc{write}$_{c,1}$ for each cell $c$, and based on the value of the counter and the bit we want to write, only one of these states gets activated. This state then copies the desired bit value to the corresponding stable bit gadget of the tape.

Note that this more sophisticated simulation method still requires a separate state for each cell of the tape. However, for a state machine of $s$ states and an available tape of length $m$, this approach can be implemented with $O(s)$ banks in the state machine, $O(\log{m})$ banks in the counter, and $O(m)$ banks for the stable bit gadgets and auxiliary states of the tape cells; in contrast to this, the crude approach has a $O(s \cdot m)$ factor in the total number of banks. Even more importantly, this simulations method allows us keep the banks modeling the state machine and the banks modeling the tape separately, thus providing a much cleaner representation of the Turing machine in our systems.

\section{A note on further possible sequential models} \label{App:modelling}

In the paper, we have studied two different sequential models of our financial networks: the reversible model (which allows banks to return from a default if they acquire new assets) and the monotone model (which guarantees an eventual stabilization for any ordering in any network). Since both the reversibility of defaults and an eventual stabilization of the network are realistic properties in real-world financial systems, it is a natural idea to try to obtain an even more accurate sequential model by appropriately combining these two settings. For example, the financial framework could ensure that the liability freezing rule is applied on a bank if, for example, it already has to report a default for the third time. Alternatively, a financial authority could actively monitor the network to look for infinite cycling patterns, and enforce liability freezing in specifically chosen situations. We leave it to future work to explore a more complex (and possibly more realistic) line of models in this direction.

Note that our discussion of sequential models is not exhaustive; there are various further changes we can execute to model the sequential process slightly differently. While these alternative models may come with some convenient properties, the changes often also introduce new undesired side effects into the model.

For a simple example of such an alternative model, assume that the financial industry is aware of the heavily dynamic behavior of payment obligations on CDSs in these networks, and thus the authorities decide to measure the assets of a bank by estimating an expected incoming payment on CDSs. That is, we select a global constant $\mu \in [0,1]$, and given a CDS of weight $\delta$ from $u$ to $v$ (in reference to $w$), we define the incoming assets of bank $v$ on this CDS as $\mu \cdot \delta \cdot r_u$, regardless of the (dynamically changing an possibly inaccurate) current value of $r_w$. 

One clear disadvantage of this slightly simpler model is that defaults can easily remain undetected in the system. E.g. in the branching gadget of Figure \ref{fig:branch}, a choice of $\mu=1$ will mean that neither of the two banks ever report a default, since they can both fulfill their obligations in a possible best-case situation. On the other hand, the survival of both $u$ and $v$ can never be an equilibrium, since neither of them receives any assets, and thus it remains undetected that either $u$ or $v$ should be in default in any `reasonable' outcome. Similarly, a choice of $\mu=0$ will force both $u$ and $v$ to report a default with $r_u=r_v=0$, even though this is not a realistic outcome in the network. While the choice of $\mu \in (0,1)$ seems like a reasonable compromise, it can actually lead to both of these problems (both undetected and false defaults) in practice, and as such, it does not allow an accurate analysis of the network.

\end{appendices}

\end{document}